\documentclass[12pt,draftclsnofoot, onecolumn]{IEEEtran}
\usepackage{amsmath,subfigure,multirow}
\usepackage{graphicx,amssymb,lineno,bm,subfigure}
\usepackage{algorithm,color}
\usepackage{algorithmic}
\usepackage{cite}
\usepackage{array}
\usepackage{float}
\usepackage{subeqnarray}
\usepackage{cases}
\usepackage{url}

\newtheorem{thm}{\textbf{Theorem}}
\newtheorem{rmk}{\textbf{Remark}}
\newtheorem{lma}{\textbf{Lemma}}
\newtheorem{defi}{\textbf{Definition}}
\newtheorem{prop}{\textbf{Proposition}}
\newtheorem{corol}{\textbf{Corollary}}

\begin{document}

% paper title
\title{Dynamic Computation Offloading for Mobile-Edge Computing with Energy Harvesting Devices\\}

\author{\IEEEauthorblockN{Yuyi Mao, Jun Zhang, and Khaled B. Letaief, \emph{Fellow, IEEE}\\ }
\thanks{The authors are with the Department of Electronic and Computer Engineering, the Hong Kong University
of Science and Technology, Clear Water Bay, Kowloon, Hong Kong (e-mail: \{ymaoac, eejzhang, eekhaled\}@ust.hk). Khaled B. Letaief is also with Hamad bin Khalifa University, Doha, Qatar (e-mail: kletaief@hkbu.edu.qa).}
\thanks{This work is supported by the Hong Kong Research Grants Council under Grant No. 16200214.}

%\IEEEauthorblockA{Dept. of ECE, The Hong Kong University of Science and Technology\\
%Email: \{ymaoac, eejzhang, eekhaled\}@ust.hk}
}

%\markboth{Extended version}%
%{Shell \MakeLowercase{\textit{et al.}}: Bare Demo of IEEEtran.cls for Journals}
\maketitle
\vspace{-45pt}
\begin{abstract}
Mobile-edge computing (MEC) is an emerging paradigm to meet the ever-increasing computation demands from mobile applications. By offloading the computationally intensive workloads to the MEC server, the quality of computation experience, e.g., the execution latency, could be greatly improved. Nevertheless, as the on-device battery capacities are limited, computation would be interrupted when the battery energy runs out. To provide satisfactory computation performance as well as achieving green computing, it is of significant importance to seek renewable energy sources to power mobile devices via energy harvesting (EH) technologies. In this paper, we will investigate a green MEC system with EH devices and develop an effective computation offloading strategy. The \emph{execution cost}, which addresses both the execution latency and task failure, is adopted as the performance metric. A low-complexity online algorithm, namely, the \emph{Lyapunov optimization-based dynamic computation offloading} (LODCO) algorithm is proposed, which jointly decides the offloading decision, the CPU-cycle frequencies for mobile execution, and the transmit power for computation offloading. A unique advantage of this algorithm is that the decisions depend only on the instantaneous side information without requiring distribution information of the computation task request, the wireless channel, and EH processes. The implementation of the algorithm only requires to solve a deterministic problem in each time slot, for which the optimal solution can be obtained either in closed form or by bisection search. Moreover, the proposed algorithm is shown to be asymptotically optimal via rigorous analysis. Sample simulation results shall be presented to verify the theoretical analysis as well as validate the effectiveness of the proposed algorithm.
\end{abstract}
\vspace{-20pt}

\section{Introduction}
The growing popularity of mobile devices, such as smart phones, tablet computers and wearable devices, is accelerating the advent of the Internet of Things (IoT) and triggering a revolution of mobile applications \cite{Gubbi1309}. With the support of on-device cameras and embedded sensors, new applications with advanced features, e.g., navigation, face recognition and interactive online gaming, have been created. Nevertheless, the tension between resource-limited devices and computation-intensive applications becomes the bottleneck for providing satisfactory quality of experience (QoE) and hence may defer the advent of a mature mobile application market \cite{Kumar1302}.

Mobile-edge computing (MEC), which provides cloud computing capabilities within the radio access networks (RAN), offers a new paradigm to liberate the mobile devices from heavy computation workloads \cite{ETSI1409}. In conventional cloud computing systems, remote public clouds, e.g., Amazon Web Services, Google Cloud Platform and Microsoft Azure, are leveraged, and thus long latency may be incurred due to data exchange in wide area networks (WANs). In contrast, MEC has the potential to significantly reduce latency, avoid congestion and prolong the battery lifetime of mobile devices by offloading the computation tasks from the mobile devices to a physically proximal MEC server \cite{Satyanarayanan0910,Kumar1004}. Thus, lots of recent efforts have been attracted from both industry \cite{ETSI1409} and academia \cite{Barbarossa1411}.

{Unfortunately, although computation offloading is effective in exploiting the powerful computation resources at cloud servers, for conventional battery-powered devices, the computation performance may be compromised due to insufficient battery energy for task offloading, i.e., mobile applications will be terminated and mobile devices will be out of service when the battery energy is exhausted. This can possibly be overcome by using larger batteries or recharging the batteries regularly. However, using larger batteries at the mobile devices implies increased hardware cost, which is not desirable. On the other hand, recharging batteries frequently is reported as the most unfavorable characteristic of mobile phones\footnote{CNN.com, ``Battery life concerns mobile users,'' available on http://edition.cnn.com/2005/TECH/ptech/09/22/phone.study/.}, and it may even be impossible in certain application scenarios, e.g., in the wireless sensor networks (WSNs) and the IoT for surveillance where the nodes are typically hard-to-reach.} Meanwhile, the rapidly increasing energy consumption of the information and communication technology (ICT) sector also brings a strong need for green computing \cite{Lamber12}. Energy harvesting (EH) is a promising technology to resolve these issues, which can capture ambient recyclable energy, including solar radiation, wind, as well as human motion energy \cite{Sudevalayam1107}, and thus it facilitates self-sustainability and perpetual operation \cite{Ulukus15}.

By integrating EH techniques into MEC, satisfactory and sustained computation performance can be achieved. While MEC with EH devices open new possibilities for cloud computing, it also brings new design challenges. In particular, the computation offloading strategies dedicated for MEC systems with battery-powered devices cannot take full benefits of the renewable energy sources. In this paper, we will develop new design methodologies for MEC systems with EH devices.

\subsection{Related Works}
Computation offloading for mobile cloud computing systems has attracted significant attention in recent years. To increase the batteries' lifetime and improve the computation performance, various code offloading frameworks, e.g., MAUI \cite{Cuervo10} and ThinkAir \cite{Kosta1203}, were proposed. However, the efficiency of computation offloading highly depends on the wireless channel condition, as the implementation of computation offloading requires data transmission. This calls for computation offloading policies that incorporate the characteristics of wireless channels \cite{DHuang1206,Munoz1510,XChen1504}. In \cite{DHuang1206}, a stochastic control algorithm adapted to the time-varying wireless channel was proposed, which determines the offloaded software components. For the femto-cloud computing systems, where the cloud server is formed by a set of femto access points, the transmit power, precoder and computation load distribution were jointly optimized in \cite{Munoz1510}. In addition, a game-theoretic decentralized computation offloading algorithm was proposed for multi-user mobile cloud computing systems \cite{XChen1504}. Nevertheless, these works assume non-adjustable processing capabilities of the central processing units (CPUs) at the mobile devices, which is not energy-efficient since the CPU energy consumption increases super-linearly with the CPU-cycle frequency \cite{Burd96}. With dynamic voltage and frequency scaling (DVFS) techniques, the local execution energy consumption for applications with strict deadline constraints is minimized by controlling the CPU-cycle frequencies \cite{WZhang1309}. Besides, a joint allocation of communication and computation resources for multi-cell MIMO cloud computing systems was proposed in \cite{Sardellitti1506}. Most recently, the energy-delay tradeoff of mobile cloud systems with heterogeneous types of computation tasks were investigated by a Lyapunov optimization algorithm, which decides the offloading policy, task allocation, CPU clock speeds and selected network interfaces \cite{JKwak1512}.

Energy harvesting was introduced to communication systems for its potential to realize self-sustainable and green communications \cite{GPiro13,YMao1506}. With non-causal side information (SI)\footnote{`Causal SI' refers to the case that, at any time instant, only the past and current SI is known, while non-causal SI means that the future SI is also available.}, including the channel side information (CSI) and energy side information (ESI), the maximum throughput of point-to-point EH fading channels can be achieved by the directional water-filling algorithm \cite{Ozel11}. The study was later extended to EH networks with causal SI \cite{LHuang13}. Cellular networks with renewable energy supplies have also been widely investigated. Resources allocation policies that maximize the energy efficiency in OFDMA systems with hybrid energy supplies (HES), i.e., both grid and harvested energy are accessible to base stations, were proposed in \cite{DNg13}. To save the grid energy consumption, a sleep control scheme for cellular networks with HES was developed in \cite{JGong14}, and a low-complexity online base station assignment and power control algorithm based on Lyapunov optimization was proposed in \cite{YMao1512}.

The design principles for MEC systems with EH devices are different from those for EH communication systems or MEC systems with battery-powered devices. On one hand, compared to EH communication systems, computation offloading policies require a joint design of the offloading decision, i.e., whether to offload a task, the CPU-cycle frequencies for mobile execution\footnote{We use ``local execution'' and ``mobile execution'' interchangeably in this paper.}, and the transmission policy for task offloading, which makes it much more challenging. On the other hand, compared to MEC systems with battery-powered devices, the design objective is shifted from minimizing the battery energy consumption to optimizing the computation performance as the harvested energy comes for free. In addition, taking care of the ESI is a new design consideration, and the time-correlated battery energy dynamics poses another challenge.

\subsection{Contributions}
In this paper, we will investigate MEC systems with EH devices and develop an effective dynamic computation offloading algorithm. Our major contributions are summarized as follows:

\begin{itemize}
\item We consider an EH device served by an MEC server, where the computation tasks can be executed locally at the device or be offloaded to the MEC server for cloud execution\footnote{It is worthwhile to point out that powering mobile devices in MEC systems with wireless energy harvesting was proposed in \cite{CYou1606}, where the harvested energy is radiated from a hybrid access point and fully controllable. This is different from the system considered in this paper where the EH process is random and uncontrollable.}. An execution cost that incorporates the execution delay and task failure is adopted as the performance metric, while DVFS and power control are adopted to optimize the mobile execution process and data transmission for computation offloading, respectively.

\item The \emph{execution cost minimization} (ECM) problem, which is an intractable high-dimensional Markov decision problem, is formulated assuming causal SI, and a low-complexity online \textbf{L}yapunov \textbf{o}ptimization-based \textbf{d}ynamic \textbf{c}omputation \textbf{o}ffloading (LODCO) algorithm is proposed. In each time slot, the system operation, including the offloading decision, the CPU-cycle frequencies for mobile execution, and the transmit power for computation offloading, only depends on the optimal solution of a deterministic optimization problem, which can be obtained either in closed form or by bisection search.

\item We identify a non-decreasing property of the scheduled CPU-cycle frequencies (the transmit power) with respect to the battery energy level, which shows that a larger amount of available energy leads to a shorter execution delay for mobile execution (MEC server execution). Performance analysis for the LODCO algorithm is also conducted. It is shown that the proposed algorithm can achieve asymptotically optimal performance of the ECM problem by tuning a two-tuple control parameters. Moreover, it does not require statistical information of the involved stochastic processes, including the computation task request, the wireless channel, and EH processes, which makes it applicable even in unpredictable environments.

\item Simulation results are provided to verify the theoretical analysis, especially the asymptotic optimality of the LODCO algorithm. Moreover, the effectiveness of the proposed policy is demonstrated by comparisons with three benchmark polices with greedy harvested energy allocation. It is shown that the LODCO algorithm not only achieves significant performance improvement in terms of execution cost, but also effectively reduces task failure.

\end{itemize}

The organization of this paper is as follows. In Section II, we introduce the system model. The ECM problem is formulated in Section III. The LODCO algorithm for the ECM problem is proposed in Section IV and its performance analysis is conducted in Section V. We show the simulation results in
Section VI and conclude this paper in Section VII.

\section{System Model}

In this section, we will introduce the system model studied in this paper, i.e., a mobile-edge computing (MEC) system with an EH device. Both the computation model and energy harvesting model will be discussed.

\subsection{Mobile-edge Computing Systems with EH Devices}

\begin{figure}[h]
\begin{center}
   \includegraphics[width=0.65\textwidth]{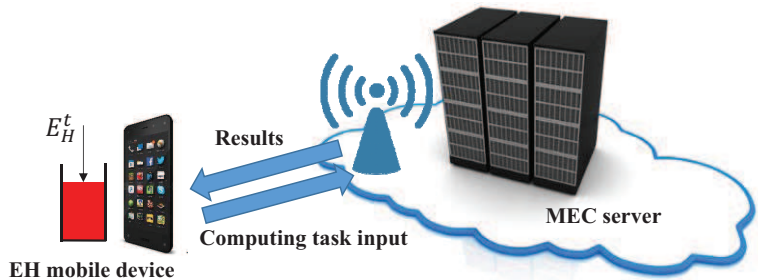}
\end{center}
\vspace{-25pt}
\caption{A mobile-edge computing system with an EH mobile device.}
\label{sysmodel}
\end{figure}

We consider an MEC system consisting of a mobile device and an MEC server as shown in Fig. \ref{sysmodel}. In particular, the mobile device is equipped with an EH component and powered purely by the harvested renewable energy. The MEC server, which could be a small data center managed by the telecom operator, is located at a distance of $d$ meters away and can be accessed by the mobile device through the wireless channel. The mobile device is associated with a system-level clone at the MEC server, namely, the cloud clone, which runs a virtual machine and can execute the computation tasks on behalf of the mobile device \cite{WZhang1309}. By offloading the computation tasks for MEC, the computation experience can be improved significantly \cite{Satyanarayanan0910,Kumar1004,Barbarossa1411}.

We assume that time is slotted, and denote the time slot length and the time slot index set by $\tau$ and $\mathcal{T}\triangleq \{0,1,\cdots\}$, respectively. The wireless channel is assumed to be independent and identically distributed (i.i.d.) block fading, i.e., the channel remains static within each time slot, but varies among different time slots. Denote the channel power gain at the $t$th time slot as $h^{t}$, and $h^{t}\sim F_{H}\left(x\right), t\in\mathcal{T}$, where $F_{H}\left(x\right)$ is the cumulative distribution function (CDF) of $h^{t}$. {For ease of reference, we list the key notations of our system model in Table \ref{notationtable}.}

{
\begin{table}[ht]
\center \protect
\caption{Summary of Key Notations}
\begin{tabular}{ll}
\hline
{\textbf{Notation}} & {\textbf{Description}}  \tabularnewline
\hline
{$d$} &Distance between the mobile device and the MEC server  \tabularnewline
{$\mathcal{T}$} &Index set of the time slots \tabularnewline
{$h^{t}$} &Channel power gain from the mobile device to the MEC server in time slot $t$\tabularnewline
{$A\left(L,\tau_{d}\right)$} &Computation task with $L$ bits input and deadline $\tau_{d}$  \tabularnewline
{$\{I_{j}^{t}\}$} &Computation mode indicators at time slot $t$ \tabularnewline
{$\zeta^{t}$} &Task arrival indicator at time slot $t$ \tabularnewline
{$X$ ($W$)} &Number of CPU cycles required to process one bit task input ($A\left(L,\tau_{d}\right)$) \tabularnewline
{$\{f_{w}^{t}\}$} &Scheduled CPU-cycle frequencies for local execution at time slot $t$  \tabularnewline
{$p^{t}$} &Transmit power for computation offloading at time slot $t$ \tabularnewline
{$f^{\max}_{\rm{CPU}}$ ($p_{\rm{tx}}^{\max}$)} &Maximum allowable CPU-cycle frequency (transmit power) \tabularnewline
{$D^{t}_{\rm{mobile}}$ ($D^{t}_{\rm{server}}$)} &Execution delay of local execution (MEC server execution) at time slot $t$\tabularnewline
{$E^{t}_{\rm{mobile}}$ ($E^{t}_{\rm{server}}$)} &Energy consumption of local execution (MEC server execution) at time slot $t$ \tabularnewline
{$e^{t}$ ($E_{H}^{t}$)} &Harvested (harvestable) energy at time slot $t$ \tabularnewline
{$E_{H}^{\max}$} &Maximum value of $E_{H}^{t}$ \tabularnewline
{$B^{t}$}  &Battery energy level at the beginning of time slot $t$ \tabularnewline
{$\phi$} &The weight of the task dropping cost \tabularnewline
\hline
\end{tabular}
\label{notationtable}
\end{table}
}

\subsection{Computation Model}
We use $A\left(L,\tau_{d}\right)$ to represent a computation task, where $L$ (in bits) is the input size of the task, and $\tau_{d}$ is the execution deadline, i.e., if it is decided that task $A\left(L,\tau_{d}\right)$ is to be executed, it should be completed within time $\tau_{d}$. The computation tasks requested by the applications running at the mobile device are modeled as an i.i.d. Bernoulli process. Specifically, at the beginning of each time slot, a computation task $A\left(L,\tau_{d}\right)$ is requested with probability $\rho$, and with probability $1-\rho$, there is no request. Denote $\zeta^{t}=1$ if a computation task is requested at the $t$th time slot and $\zeta^{t}=0$ if otherwise, i.e., $\mathbb{P}\left(\zeta^{t}=1\right)=1-\mathbb{P}\left(\zeta^{t}=0\right)=\rho, t\in\mathcal{T}$. We focus on delay-sensitive applications with execution deadline less than the time slot length, i.e., $\tau_{d}\leq \tau$ \cite{DHuang1206,Munoz1510,XChen1504,Sardellitti1506,ZJiang1512}, and assume no buffer is available for queueing the computation requests.
%\footnote{For applications with less stringent delay requirement, i.e., $\tau_{d}>\tau$, it is possible that a computing task experiences multiple channel %blocks when it is being offloaded to the cloud server, which makes the computation offloading policy design much more challenging as two-timescale %stochastic optimization will be involved. This will be investigated in our future works.}.

Each computation task can either be executed locally at the mobile device, or be offloaded to and executed by the MEC server. It may also happen that neither of these two computation modes is feasible, e.g., when energy is insufficient at the mobile device, and hence the computation task will be dropped. Denote $I_{j}^{t}\in\{0,1\}$ with $j=\{\rm{m},\rm{s},\rm{d}\}$ as the computation mode indicators, where $I_{\rm{m}}^{t}=1$ and $I_{\rm{s}}^{t}=1$ indicate that the computation task requested in the $t$th time slot is executed at the mobile device and offloaded to the MEC server, respectively, while $I_{\rm{d}}^{t}=1$ means the computation task is dropped. Thus, the computation mode indicators should satisfy the following operation constraint:
\begin{equation}
I_{\rm{m}}^{t}+I_{\rm{s}}^{t}+I_{\rm{d}}^{t}=1,  t\in\mathcal{T}.
\label{indicatorconstraint}
\end{equation}

\textbf{Local Executing Model:} The number of CPU cycles required to process one bit input is denoted as $X$, which varies from different applications and can be obtained through off-line measurement \cite{Miettinen10}. In other words, $W=LX$ CPU cycles are needed in order to successfully execute task $A\left(L,\tau_{d}\right)$. The frequencies scheduled for the $W$ CPU cycles in the $t$th time slot are denoted as $f_{w}^{t},w=1,\cdots,W$, which can be implemented by adjusting the chip voltage with DVFS techniques \cite{Rabaey96}. As a result, the delay for executing the computation task requested in the $t$th time slot locally at the mobile device can be expressed as
\begin{equation}
D^{t}_{\rm{mobile}}= \sum_{w=1}^{W}\left(f_{w}^{t}\right)^{-1}.
\label{delaymobile}
\end{equation}
Accordingly, the energy consumption for local execution by the mobile device is given by
\begin{equation}
E^{t}_{\rm{mobile}}= \kappa \sum_{w=1}^{W}\left(f_{w}^{t}\right)^{2},
\label{energymobile}
\end{equation}
where $\kappa$ is the effective switched capacitance that depends on the chip architecture \cite{Burd96}. Moreover, we assume the CPU-cycle frequencies are constrained by $f_{\rm{CPU}}^{\max}$, i.e., $f_{w}^{t}\leq f_{\rm{CPU}}^{\max},\forall w$.

\textbf{Mobile-edge Executing Model:} In order to offload the computation task for MEC, the input bits of $A\left(L,\tau_{d}\right)$ should be transmitted to the MEC server. {We assume sufficient computation resource, e.g., a high-speed multi-core CPU, is available at the MEC server, and thus ignore its execution delay \cite{WZhang1309,JKwak1512,CYou1606}.} It is further assumed that the output of the computation is of small size so the transmission delay for feedback is negligible. Denote the transmit power as $p^{t}$, which should be less than the maximum transmit power $p_{\rm{tx}}^{\max}$. According to the Shannon-Hartley formula, the achievable rate in the $t$th time slot is given by
%\begin{equation}
$r\left(h^{t},p^{t}\right) = \omega \log_{2} \left(1+ \frac{h^{t}p^{t}}{\sigma}\right)$,
%\label{rate}
%\end{equation}
where $\omega$ is the system bandwidth and $\sigma$ is the noise power at the receiver. Consequently, if the computation task is executed by the MEC server, the execution delay equals the transmission delay for the input bits, i.e.,
\begin{equation}
D^{t}_{\rm{server}}= \frac{L}{r\left(h^{t},p^{t}\right)},
\label{delaycloud}
\end{equation}
and\footnote{When the execution delay in the MEC server is non-negligible, the proposed algorithm can still be applied by modifying the expression of $D^{t}_{\rm{server}}$ in (\ref{delaycloud}) as $D^{t}_{\rm{server}}=L\slash r\left(h^{t},p^{t}\right)+\tau_{\rm{server}}$, where $\tau_{\rm{server}}$ denotes the execution delay in the MEC server.} the energy consumed by the mobile device is given by
\begin{equation}
E^{t}_{\rm{server}}= p^{t}\cdot D^{t}_{\rm{server}} = p^{t}\cdot \frac{L}{r\left(h^{t},p^{t}\right)}.
\label{energycloud}
\end{equation}

\subsection{Energy Harvesting Model}
The EH process is modeled as successive energy packet arrivals, i.e., $E_{H}^{t}$ units of energy arrive at the mobile device at the beginning of the $t$th time slot. We assume $E_{H}^{t}$'s are i.i.d. among different time slots with the maximum value of $E_{H}^{\max}$. Although the i.i.d. model is simple, it captures the stochastic and intermittent nature of the renewable energy processes \cite{LHuang13,YMao1512,Laksh14}. In each time slot, part of the arrived energy, denoted as $e^{t}$, satisfying
\begin{equation}
0\leq e^{t} \leq E_{H}^{t}, t\in\mathcal{T},
\label{harvestableconst}
\end{equation}
will be harvested and stored in a battery, and it will be available for either local execution or computation offloading starting from the next time slot. We start by assuming that the battery capacity is sufficiently large. Later we will show that by picking the values of $e^{t}$'s, the battery energy level is deterministically upper-bounded under the proposed computation offloading policy, and thus we only need a finite-capacity battery in actual implementation. More importantly, including $e^{t}$'s as optimization variables facilitates the derivation and performance analysis of the proposed algorithm. Similar techniques were adopted in previous studies, such as \cite{LHuang13}, \cite{YMao1512} and \cite{Laksh14}. Denote the battery energy level at the beginning of time slot $t$ as $B^{t}$. Without loss of generality, we assume $B^{0}=0$ and $B^{t}<+\infty, t\in\mathcal{T}$. In this paper, energy consumed for purposes other than local computation and transmission is ignored for simplicity, while more general energy models can be handled by the proposed algorithm with minor modifications.{\footnote{We will demonstrate how to adapt the proposed algorithm to more general energy models of mobile devices, e.g., by taking the power consumption of screens and operating systems into account, in Section IV-A.}} Denote the energy consumed by the mobile device in time slot $t$ as $\mathcal{E}\left(\bm{I}^{t},\bm{f}^{t},p^{t}\right)$, which depends on the selected computation mode, scheduled CPU-cycle frequencies and transmit power, and can be expressed as
\begin{equation}
\mathcal{E}\left(\bm{I}^{t},\bm{f}^{t},p^{t}\right)=I_{\rm{m}}^{t}E^{t}_{\rm{mobile}} + I_{\rm{s}}^{t} E^{t}_{\rm{server}},
\end{equation}
subject to the following energy causality constraint:
\begin{equation}
\mathcal{E}\left(\bm{I}^{t},\bm{f}^{t},p^{t}\right)\leq B^{t}<+\infty, t\in\mathcal{T}.
\label{EHcausality}
\end{equation}
Thus, the battery energy level evolves according to the following equation:
\begin{equation}
B^{t+1}=B^{t}-\mathcal{E}\left(\bm{I}^{t},\bm{f}^{t},p^{t}\right)+e^{t}, t\in\mathcal{T}.
\label{batterydynamics}
\end{equation}

%{\color{blue}This can be realized by deploying a battery management system (BMS) at the mobile device in practice \cite{Eichi1306}.}

With EH mobile devices, the computation offloading policy design for MEC systems becomes much more complicated compared to that of conventional mobile cloud computing systems with battery-powered devices. Specifically, both the ESI and CSI need to be handled, and the temporally correlated battery energy level makes the system decision coupled in different time slots. Consequently, an optimal computation offloading strategy should strike a good balance between the computation performance of the current and future computation tasks.

\section{Problem Formulation}
In this section, we will first introduce the performance metric, namely, the execution cost. The execution cost minimization (ECM) problem will then be formulated and its unique technical challenges will be identified.

\subsection{Execution Cost Minimization Problem}
Execution delay is one of the key measures for users' QoE \cite{DHuang1206,Munoz1510,XChen1504,WZhang1309,Sardellitti1506,JKwak1512}, which will be adopted to optimize the computation offloading policy for the considered MEC system. Nevertheless, due to the intermittent and sporadic nature of the harvested energy, some of the requested computation tasks may not be able to be executed and have to be dropped, e.g., due to lacking of energy for local computation, while the wireless channel from the mobile device to the MEC server is in deep fading, i.e., the input of the tasks cannot be delivered. To take this aspect into consideration, we penalize each dropped task by a unit of cost. Thus, we define the execution cost as the weighted sum of the execution delay and the task dropping cost, which can be expressed by the following formula:
\begin{equation}
{\rm{cost}}^{t}=\mathcal{D}\left(\bm{I}^{t},\bm{f}^{t},p^{t}\right)+\phi \cdot \bm{1}\left(\zeta^{t}=1,I_{\rm{d}}^{t}=1\right),
\label{execost}
\end{equation}
where $\phi$ (in second) is the weight of the task dropping cost, $\bm{1}\left(\cdot\right)$ is the indicator function, and $\mathcal{D}\left(\bm{I}^{t},\bm{f}^{t},p^{t}\right)$ is given by
\begin{equation}
\mathcal{D}\left(\bm{I}^{t},\bm{f}^{t},p^{t}\right)=\bm{1}\left(\zeta^{t}=1\right)\cdot \left(I_{\rm{m}}^{t} D^{t}_{\rm{mobile}} + I_{\rm{s}}^{t} D^{t}_{\rm{server}}\right).
\end{equation}
Without loss of generality, we assume that executing a task successfully is preferred to dropping a task, i.e., $\tau_{d}\leq \phi$.

If it is decided that a task is to be executed, i.e., $I_{\rm{m}}^{t}=1$ or $I_{\rm{s}}^{t}=1$, it should be completed before the deadline $\tau_{d}$. In other words, the following deadline constraint should be met:
\begin{equation}
\mathcal{D}\left(\bm{I}^{t},\bm{f}^{t},p^{t}\right)\leq \tau_{d},t\in\mathcal{T}.
\label{deadconst}
\end{equation}
Consequently, the ECM problem can be formulated as:
\begin{align}
%\begin{split}
&\mathcal{P}_{1}:\ \min\limits_{\bm{I}^{t},\bm{f}^{t},p^{t},e^{t}}\lim\limits_{T\rightarrow \infty}\frac{1}{T}\mathbb{E}\left[\sum_{t=0}^{T-1}{\rm{cost}}^{t}\right]\nonumber\\
&\ \ \ \ \ \ \ \ \ {\rm{s.t.}}\ \ \ (\ref{indicatorconstraint}), (\ref{harvestableconst}), (\ref{EHcausality}), (\ref{deadconst})\nonumber\\
&\ \ \ \ \ \ \ \ \ \ \ \ \ \ \ \ I_{\rm{m}}^{t}+I_{\rm{s}}^{t}\leq \zeta^{t},  t\in\mathcal{T}\label{nonemptytask}\\
&\ \ \ \ \ \ \ \ \ \ \ \ \ \ \ \ \mathcal{E}\left(\bm{I}^{t},\bm{f}^{t},p^{t}\right)\leq E_{\max},  t\in\mathcal{T}\label{discharging}\\
&\ \ \ \ \ \ \ \ \ \ \ \ \ \ \ \ 0\leq p^{t} \leq p_{\rm{tx}}^{\max}\cdot \bm{1}\left(I_{\rm{s}}^{t}=1\right),  t\in\mathcal{T} \label{maxtranspwr}\\
&\ \ \ \ \ \ \ \ \ \ \ \ \ \ \ \ 0\leq f_{w}^{t}\leq f_{\max}^{\rm{CPU}}\cdot \bm{1}\left(I_{\rm{m}}^{t}=1\right), w=1,\cdots,W, t\in\mathcal{T}, \label{maxCPUfreq}\\
&\ \ \ \ \ \ \ \ \ \ \ \ \ \ \ \ I_{\rm{m}}^{t},I_{\rm{s}}^{t},I_{\rm{d}}^{t}\in\{0,1\},t\in\mathcal{T}\label{zeroone},
%\end{split}
\end{align}
where (\ref{nonemptytask}) indicates that if there is no computation task requested, neither mobile execution nor MEC server execution is feasible. (\ref{discharging}) is the battery discharging constraint, i.e., the amount of battery output energy cannot exceed $E_{\max}$ in each time slot, which is essential for preventing the battery from over discharging \cite{Laksh14,SSun1507}. The maximum allowable transmit power and the maximum CPU-cycle frequency constraints are imposed by (\ref{maxtranspwr}) and (\ref{maxCPUfreq}), respectively, while the zero-one indicator constraint for the computation mode indicators is represented by (\ref{zeroone}).

\subsection{Problem Analysis}
In the considered MEC system, the system state is composed of the task request, the harvestable energy, the battery energy level, as well as the channel state, and the action is the energy harvesting and the computation offloading decision, including the scheduled CPU-cycle frequencies and the allocated transmit power. It can be checked that the allowable action set depends only on the current system state, and is irrelevant with the state and action history. Besides, the objective is the long-term average execution cost. Thus, $\mathcal{P}_{1}$ is a Markov decision process (MDP) problem. In principle, $\mathcal{P}_{1}$ can be solved optimally by standard MDP algorithms, e.g., the \emph{relative value iteration algorithm} and the \emph{linear programming reformulation approach} \cite{BertsekasDP}. Nevertheless, for both algorithms, we need to use finite states to characterize the system, and discretize the feasible action set. For example, if we use $K=20$ states to quantize the wireless channel, $M=20$ states to characterize the battery energy level, $E=5$ states to describe the harvestable energy, and admits $L=10$ transmit power levels and $F=10$ CPU-cycle frequencies, there are $2KME=4000$ possible system states in total. For the relative value iteration algorithm, this will take a long time to converge as there will be as many as $L+1+F^{W}$ feasible actions in some states. For the linear programming (LP) reformulation approach, we need to solve an LP problem with $2KME\times \left(L+1+F^{W}\right)$ variables, which will be practically infeasible even for a small value of $W$, e.g., $1000$. In addition, it will be difficult to obtain solution insights with the MDP algorithms as they are based on numerical iteration. Moreover, quantizing the state and action may lead to severe performance degradation, and the memory requirement for storing the optimal policy will yet be another big challenge.

In the next section, we will propose a \textbf{L}yapunov \textbf{o}ptimization-based \textbf{d}ynamic \textbf{c}omputation \textbf{o}ffloading (LODCO) algorithm to solve $\mathcal{P}_{1}$, which enjoys the following favorable properties:
\begin{itemize}
    \item There is no need to quantize the system state and feasible action set, and the decision of the LODCO algorithm within each time slot is of low complexity. In addition, there is no memory requirement for storing the optimal policy.
    \item The LODCO algorithm has no prior information requirement on the channel statistics, the distribution of the renewable energy process or the computation task request process.
    \item The performance of the LODCO algorithm is controlled by a two-tuple control parameters. Theoretically, by adjusting these parameters, the proposed algorithm can behave arbitrarily close to the optimal performance of $\mathcal{P}_{1}$.
    \item  An upper bound of the required battery capacity is obtained, which shall provide guidelines for practical installation of the EH components and storage units.
\end{itemize}

\section{Dynamic Computation Offloading: The LODCO Algorithm}
In this section, we will develop the LODCO algorithm to solve $\mathcal{P}_{1}$. We will first show an important property of the optimal CPU-cycle frequencies, which helps to simplify $\mathcal{P}_{1}$. In order to take advantages of Lyapunov optimization, we will introduce a modified ECM problem to assist the algorithm design. The LODCO algorithm will be then proposed for the modified problem, which also provides a feasible solution to $\mathcal{P}_{1}$. In Section V, we will show that this solution is asymptotically optimal for $\mathcal{P}_{1}$.

\subsection{The LODCO Algorithm}
We first show that the optimal CPU-cycle frequencies of the $W$ CPU cycles scheduled for a single computation task should be the same, as stated in the following lemma.
\begin{lma}
If a task requested at the $t$th time slot is being executed locally, the optimal frequencies of the $W$ CPU cycles should be the same, i.e., $f_{w}^{t}=f^{t}, w=1,\cdots, W$.
\label{equalfreq}
\end{lma}
\begin{proof}
The proof can be obtained by contradiction, which is omitted for brevity.
%Suppose there exists an optimal $\bm{f}^{t}=\{f_{w}^{t}\}$ with $f^{t}_{i}>f^{t}_{j}, i\neq j$, we construct a new solution %$\tilde{\bm{f}}^{t}=\{\tilde{f}_{w}^{t}\}$ as
%\begin{equation}
%\tilde{f}_{w}^{t}=
%\begin{cases}
%f_{w}^{t}, &w\neq i,j\\
%2f_{i}^{t}f_{j}^{t}\cdot\left(f_{i}^{t}+f_{j}^{t}\right)^{-1}, &w=i,j,
%\end{cases}
%\end{equation}
%which achieves the same execution delay, while reduces the harvested energy consumption by
%\begin{equation}
%\begin{split}
%\Delta E &=\kappa \left(f_{i}^{t}\right)^{2} + \kappa \left(f_{j}^{t}\right)^{2} - 2 \kappa %\left(\frac{2f_{i}^{t}f_{j}^{t}}{f_{i}^{t}+f_{j}^{t}}\right)^{2}\\
%&=\kappa \left(f_{i}^{t}+f_{j}^{t}\right)^{-2}\cdot %\left[\left(\left(f_{i}^{t}\right)^{2}-\left(f_{j}^{t}\right)^{2}\right)^{2}+2f_{i}^{t}f_{j}^{t}\left(f_{i}^{t}-f_{j}^{t}\right)^{2}\right]>0.
%\end{split}
%\end{equation}
%It can be verified that $\tilde{f}_{i}^{t}=\tilde{f}_{j}^{t}<\max\{f_{i}^{t},f_{j}^{t}\}$. Thus, $\tilde{\bm{f}}^{t}$ is also feasible and optimal.
\end{proof}

The property of the optimal CPU-cycle frequencies in Lemma \ref{equalfreq} indicates that we can optimize a scalar $f^{t}$ instead of a $W$-dimensional vector $\bm{f}^{t}$ for each computation task, which helps to reduce the number of optimization variables. However, due to the energy causality constraint (\ref{EHcausality}), the system's decisions are coupled among different time slots, which makes the design challenging. This is a common difficulty for the design of EH systems. We find that by introducing a non-zero lower bound, $E_{\min}$, on the battery output energy at each time slot, such coupling effect can be eliminated and the system operations can be optimized by ignoring (\ref{EHcausality}) at each time slot. Thus, we first introduce a modified version of $\mathcal{P}_{1}$ as
\begin{align}
&\mathcal{P}_{2}:\ \min\limits_{\bm{I}^{t},f^{t},p^{t},e^{t}}\lim\limits_{T\rightarrow \infty}\frac{1}{T}\mathbb{E}\left[\sum_{t=0}^{T-1}{\rm{cost}}^{t}\right]\nonumber\\
&\ \ \ \ \ \ \ \ \ {\rm{s.t.}}\ \ \ (\ref{indicatorconstraint}), (\ref{harvestableconst}), (\ref{EHcausality}), (\ref{deadconst})-(\ref{zeroone})\nonumber\\
&\ \ \ \ \ \ \ \ \ \ \ \ \ \ \ \ \mathcal{E}\left(\bm{I}^{t},f^{t},p^{t}\right)\in \{0\}\bigcup \left[E_{\min},E_{\max}\right],t\in\mathcal{T} \label{tightenedconst},
\end{align}
where $0< E_{\min}\leq E_{\max}$. Compared to $\mathcal{P}_{1}$, only a scalar $f^{t}$ needs to be determined for mobile execution, which preserves optimality according to Lemma \ref{equalfreq}, and thus $D_{\rm{mobile}}^{t}=W\left(f^{t}\right)^{-1}$ and $E_{\rm{mobile}}^{t}=W\kappa \left(f^{t}\right)^{2}$. {Besides, all constraints in $\mathcal{P}_{1}$ are retained in $\mathcal{P}_{2}$, and an additional constraint on the battery output energy is imposed by (\ref{tightenedconst}).} Hence, $\mathcal{P}_{2}$ is a tightened version of $\mathcal{P}_{1}$. Denote the optimal values of $\mathcal{P}_{1}$ and $\mathcal{P}_{2}$ as $\rm{EC}_{\mathcal{P}_{1}}^{*}$ and $\rm{EC}_{\mathcal{P}_{2}}^{*}$, respectively. The following proposition reveals the relationship between $\rm{EC}_{\mathcal{P}_{1}}^{*}$ and $\rm{EC}_{\mathcal{P}_{2}}^{*}$, which will later help show the asymptotic optimality of the proposed algorithm.

\begin{prop}
The optimal value of $\mathcal{P}_{2}$ is greater than that of $\mathcal{P}_{1}$, but smaller than the optimal value of $\mathcal{P}_{1}$ plus a positive constant $\nu\left(E_{\min}\right)$, i.e., ${\rm{EC}}_{\mathcal{P}_{1}}^{*}\leq {\rm{EC}}_{\mathcal{P}_{2}}^{*} \leq {\rm{EC}}_{\mathcal{P}_{1}}^{*}+\nu\left(E_{\min}\right)$, where $\nu\left(E_{\min}\right)=\rho\left[\phi\left(1-F_{H}\left(\eta\right)\right)+\bm{1}_{E_{\min}\geq E_{\min}^{\tau_{d}}}\cdot \left(\phi-\tau_{E_{\min}}\right)\right]$. Here, $\eta=\left(2^{\frac{L}{\tau_{d}\omega}}-1\right)\sigma\tau_{d}E_{\min}^{-1}$, $E_{\min}^{\tau_{d}}=\kappa W^{3} \tau_{d}^{-2}$ and $\tau_{E_{\min}}=\kappa^{\frac{1}{2}}W^{\frac{3}{2}}E_{\min}^{-\frac{1}{2}}$.
\label{tightperf}
\end{prop}
\begin{proof}
Please refer to Appendix A.
\end{proof}

In general, the upper bound in Proposition \ref{tightperf} is not tight. However, as $E_{\min}$ goes to zero, $\nu\left(E_{\min}\right)$ diminishes as shown in the following corollary.
\begin{corol}
By letting $E_{\min}$ approach zero, ${\rm{EC}}_{\mathcal{P}_{2}}^{*}$ can be made arbitrarily close to ${\rm{EC}}_{\mathcal{P}_{1}}^{*}$, i.e., $\lim\limits_{E_{\min}\rightarrow 0}\nu\left(E_{\min}\right)=0$.
\label{asym1}
\end{corol}
%\begin{proof}
%First, since when $E_{\min}< E_{\min}^{\tau_{d}}$, $\bm{1}_{E_{\min}\geq E_{\min}^{\tau_{d}}}\cdot \left(\phi-\tau_{E_{\min}}\right)=0$, we have %$\lim\limits_{E_{\min}\rightarrow 0}\bm{1}_{E_{\min}\geq E_{\min}^{\tau_{d}}}\cdot \left(\phi-\tau_{E_{\min}}\right)=0$. Meanwhile, since $\lim\limits %_{E_{\min}\rightarrow 0}\eta = +\infty$ and $\lim\limits_{x\rightarrow +\infty} F_{H}\left(x\right)=1$, we have $\lim\limits_{E_{\min}\rightarrow 0} %F_{H}\left(\eta\right)=1$, i.e., $\lim\limits_{E_{\min}\rightarrow 0} 1-F_{H}\left(\eta\right)=0$. Thus, the desired result is obtained.
%\end{proof}
\begin{proof}
The proof is omitted due to space limitation.
\end{proof}

Proposition \ref{tightperf} bounds the optimal performance of $\mathcal{P}_{2}$ by that of $\mathcal{P}_{1}$, while Corollary \ref{asym1} shows that the performance of both problems can be made arbitrarily close. Actually, Corollary \ref{asym1} fits our intuition, since when $E_{\min}\rightarrow 0$, $\mathcal{P}_{2}$ reduces to $\mathcal{P}_{1}$. However, due to the temporally correlated battery energy levels, the system's decisions are time-dependent, and thus the vanilla version of Lyapunov optimization techniques, where the allowable action sets are i.i.d., cannot be applied directly. Fortunately, the weighted perturbation method offers an effective solution to circumvent this issue \cite{Neely10CDC}.
In order to present the algorithm, we first define the perturbation parameter and the virtual energy queue at the mobile device, which are two critical elements.
\begin{defi}
The perturbation parameter $\theta$ for the EH mobile device is a bounded constant satisfying
\begin{equation}
\theta \geq \tilde{E}_{\max}+ V \phi \cdot  E_{\min}^{-1},
\label{perturbpara}
\end{equation}
where $\tilde{E}_{\max}=\min\{\max\{\kappa W \left(f_{\rm{CPU}}^{\max}\right)^{2},p_{\rm{tx}}^{\max}\tau\},E_{\max}\}$, and $0<V<+\infty$ is a control parameter in the LODCO algorithm with unit as ${\rm{J}^{2}}\cdot {\rm{second}}^{-1}$.{\footnote{Since the right-hand side of (\ref{perturbpara}) increases with $\phi$ ($\phi\in\left[\tau_{d},+\infty\right)$), a larger value of $\phi$ will result in a large value of $\theta$, i.e., a higher perturbed energy level in the proposed algorithm.}}
\end{defi}

\begin{defi}
The virtual energy queue $\tilde{B}^{t}$ is defined as
%\begin{equation}
$\tilde{B}^{t}=B^{t}-\theta$,
%\label{virtuequeue}
%\end{equation}
which is a shifted version of the actual battery energy level at the mobile device.
\label{virtuequeue}
\end{defi}

As will be elaborated later, the proposed algorithm minimizes the weighted sum of the net harvested energy and the execution cost in each time slot, with weights of the virtual energy queue length $\tilde{B}^{t}$, and the control parameter $V$, respectively, which tends to stabilize $B^{t}$ around $\theta$ and meanwhile minimize the execution cost. The LODCO algorithm is summarized in Algorithm \ref{alg1}. In each time slot, the system operation is determined by solving a deterministic per-time slot problem, which is parameterized by the current system state and with all constraints in $\mathcal{P}_{2}$ except the energy causality constraint (\ref{EHcausality}).

{
\begin{rmk}
When the power consumption for maintaining the basic operations at the mobile device, denoted as $P_{\rm{basic}}$, is considered, there will be four computation modes for the time slots with $\zeta^{t}=1$, i.e., mobile execution ($I_{\rm{m}}^{t}=1$), MEC server execution ($I_{\rm{s}}^{t}=1$), dropping the task while maintaining the basic operations ($I_{\rm{d}}^{t}=1$), as well as dropping the task and disabling the basic operations ($I_{\rm{f}}^{t}=1$); while for the time slots with $\zeta^{t}=0$, two modes exist, i.e., the basic operations are maintained ($I_{\rm{d}}^{t}=1$) or disabled ($I_{\rm{f}}^{t}=1$). As a result, the energy consumed by the mobile device at the $t$th time slot can be written as $\mathcal{E}\left(\bm{I}^{t},f^{t},p^{t}\right)=I_{\rm{m}}^{t}E^{t}_{\rm{mobile}}+I_{\rm{s}}^{t}E^{t}_{\rm{server}}+\left(I_{\rm{m}}^{t}+I_{\rm{s}}^{t}+I_{\rm{d}}^{t}\right)P_{\rm{basic}}\tau$.
We introduce a unit of cost to penalize the interruption of basic operations, and thus the execution cost can be expressed as ${\rm{cost}}^{t}=\mathcal{D}\left(\bm{I}^{t},\bm{f}^{t},p^{t}\right)+\phi \cdot \bm{1}\left(\zeta^{t}=1,I_{\rm{d}}^{t}\ {\rm{or}}\ I_{\rm{f}}^{t}=1\right)+{\psi} \cdot \bm{1}\left(I_{\rm{f}}^{t}=1\right)$, where ${\psi}>0$ is the weight of the basic operations interruption cost. It is worthwhile to note that the framework of the proposed LODCO algorithm can be modified for this case, where the major changes lie on the selection of the perturbation parameter $\theta$ and the solution for the per-time slot problem, and will not be detailed in this paper.
\end{rmk}}

{\color{black}
\begin{algorithm}[h]
\caption{The LODCO Algorithm}
\label{alg1}
\begin{algorithmic}[1]
\STATE At the beginning of time slot $t$, obtain the task request indicator $\zeta^{t}$, virtual energy queue length $\tilde{B}^{t}$, harvestable energy $E_{H}^{t}$, and channel gain $h^{t}$.
\STATE Decide $e^{t}, \bm{I}^{t}$, $f^{t}$ and $p^{t}$ by solving the following deterministic problem:
\begin{align}
&\min_{\bm{I}^{t},p^{t}, f^{t}, e^{t}} \tilde{B}^{t}\left[e^{t}-\mathcal{E}\left(\bm{I}^{t},f^{t},p^{t}\right)\right]+V\left[\mathcal{D}\left(\bm{I}^{t},f^{t},p^{t}\right)+\phi\cdot \bm{1}\left(\zeta^{t}=1, I_{\rm{d}}^{t}=1\right)\right]\nonumber\\
&\ \ \ \mathrm{s.t.}\ \ \ (\ref{indicatorconstraint}), (\ref{harvestableconst}), (\ref{deadconst})-(\ref{tightenedconst})\nonumber.
\end{align}
\STATE Update the virtual energy queue according to (\ref{batterydynamics}) and Definition \ref{virtuequeue}.
\STATE Set $t=t+1$.
\end{algorithmic}
\end{algorithm}}

\subsection{Optimal Computation Offloading in Each Time Slot}
In this subsection, we will develop the optimal solution for the per-time slot problem, which consists of two components: the optimal energy harvesting, i.e., to determine $e^{t}$, as well as the optimal computation offloading decision, i.e., to determine $\bm{I}^{t}$, $f^{t}$ and $p^{t}$. The results obtained in this subsection are essential for feasibility verification and performance analysis of the LODCO algorithm in Section V.

\textbf{Optimal Energy Harvesting:} It is straightforward to show that the optimal amount of harvested energy $e^{t*}$ can be obtained by
solving the following LP problem:
\begin{equation}
\min_{0\leq e^{t}\leq E_{H}^{t}} \tilde{B}^{t}e^{t},
\end{equation}
and its optimal solution is given by
\begin{equation}
e^{t*}=E_{H}^{t}\cdot \bm{1}\{\tilde{B}^{t}\leq 0\}.
\label{optEH}
\end{equation}

\textbf{Optimal Computation Offloading:}
After decoupling $e^{t}$ from the objective function, we can then simplify the per-time slot problem into the following optimization problem $\mathcal{P}_{\rm{CO}}$:
\begin{equation}
\begin{split}
&\mathcal{P}_{\rm{CO}}:\min_{\bm{I}^{t},f^{t},p^{t}} -\tilde{B}^{t}\cdot \mathcal{E}\left(\bm{I}^{t},f^{t},p^{t}\right)+V\left[(\mathcal{D}\left(\bm{I}^{t},f^{t},p^{t}\right)+\phi\cdot \bm{1}\left(\zeta^{t}=1, I_{\rm{d}}^{t}=1\right)\right]\\
&\ \ \ \ \ \ \ \ \ {\rm{s.t.}}\ \ (\ref{indicatorconstraint}), (\ref{deadconst})-(\ref{tightenedconst}).
\end{split}
\end{equation}
Denote the feasible action set and the objective function of $\mathcal{P}_{\rm{CO}}$ as $\mathcal{F}^{t}_{\rm{CO}}$ and $J^{t}_{\rm{CO}}\left(\bm{I}^{t},f^{t},p^{t}\right)$, respectively. For the time slots without computation task request, i.e., $\zeta^{t}=0$, there is a single feasible solution for $\mathcal{P}_{\rm{CO}}$ due to (\ref{nonemptytask}), which is given by $I_{\rm{m}}^{t}=I_{\rm{s}}^{t}=0$, $I^{t}_{\rm{d}}=1$, $f^{t}=0$, and $p^{t}=0$. Thus, we will focus on the time slots with computation task requests in the following. First, we obtain the optimal CPU-cycle frequency for a task being executed locally at the mobile device by solving the following optimization problem $\mathcal{P}_{\rm{ME}}$:
\begin{align}
&\mathcal{P}_{\rm{ME}}: \min_{f^{t}} -\tilde{B}^{t}\cdot \kappa W \left(f^{t}\right)^{2}+V\cdot \frac{W}{f^{t}}\nonumber\\
&\ \ \ \ \ \ \ \ {\rm{s.t.}}\ \ \ 0<f^{t}\leq f_{\rm{CPU}}^{\max}\label{maxCPUfreq2}\\
&\ \ \ \ \ \ \ \ \ \ \ \ \ \ \frac{W}{f^{t}}\leq \tau_{d} \label{delayME}\\
&\ \ \ \ \ \ \ \ \ \ \ \ \ \ \kappa W \left(f^{t}\right)^{2}\in\left[E_{\min},E_{\max}\right] \label{energyME},
\end{align}
which is obtained by plugging $I_{\rm{m}}^{t}=1$, $I_{\rm{s}}^{t}=I_{\rm{d}}^{t}=0$ and $p^{t}=0$ into {$\mathcal{P}_{\rm{CO}}$}, and using the fact that $f^{t}>0$ for local execution. (\ref{delayME}) is the execution delay constraint for mobile execution, and (\ref{energyME}) is the CPU energy consumption constraint obtained by combining (\ref{discharging}) and (\ref{tightenedconst}). We denote the objective function of $\mathcal{P}_{\rm{ME}}$ as $J_{\rm{m}}^{t}\left(f^{t}\right)$. Note that mobile execution is not necessarily feasible due to limited computation capability of the processing unit at the mobile device as indicated by (\ref{maxCPUfreq2}). In the following proposition, we develop the feasibility condition and the optimal solution for $\mathcal{P}_{\rm{ME}}$ given it is feasible.
\begin{prop}
$\mathcal{P}_{\rm{ME}}$ is feasible if and only if $f_{L}\leq f_{U}$, where $f_{L}=\max\{\sqrt{\frac{E_{\min}}{\kappa W}},\frac{W}{\tau_{d}}\}$ and $f_{U}=\min\{\sqrt{\frac{E_{\max}}{\kappa W}},f_{\max}\}$. If $\mathcal{P}_{\rm{ME}}$ is feasible, its optimal solution is given by:
\begin{equation}
f^{t*}=\begin{cases}
f_{U}, &\tilde{B}^{t}\geq 0\ {\rm{or}}\ \tilde{B}^{t}<0, f_{0}^{t}> f_{U}\\
f_{0}^{t}, &\tilde{B}^{t}<0, f_{L}\leq f_{0}^{t} \leq f_{U}\\
f_{L}, &\tilde{B}^{t}<0, f_{0}^{t}<f_{L},
\end{cases}
\label{optCPU}
\end{equation}
where $f^{t}_{0}=\left(\frac{V}{-2\tilde{B}^{t}\kappa}\right)^{\frac{1}{3}}$.
\label{optCPUlma}
\end{prop}
\begin{proof}
We first show the feasibility condition. Due to (\ref{delayME}), $f^{t}$ should be no less than $W\slash \tau_{d}$ in order to meet the delay constraint. Besides, since the CPU energy consumption increases with $f^{t}$, the battery output energy constraint can be equivalently expressed as $\sqrt{\frac{E_{\min}}{\kappa W}} \leq f^{t}\leq \sqrt{\frac{E_{\max}}{\kappa W}}$. By incorporating (\ref{maxCPUfreq2}), we rewrite the feasible CPU-cycle frequency set as $f_{L}=\max\{\sqrt{\frac{E_{\min}}{\kappa W}},W\slash \tau_{d}\}\leq f^{t}\leq f_{U}=\min\{\sqrt{\frac{E_{\max}}{\kappa W}},f_{\max}\}$, i.e., $\mathcal{P}_{\rm{ME}}$ is feasible if and only if $f_{L}\leq f_{U}$.

Next, we proceed to show the optimality of (\ref{optCPU}) when $\mathcal{P}_{\rm{ME}}$ is feasible. When $\tilde{B}^{t}\geq 0$, $J_{\rm{m}}^{t}\left(f^{t}\right)$ decreases with $f^{t}$, i.e., the minimum value is achieved by $f^{t}=f_{U}$. When $\tilde{B}^{t}<0$, $J_{\rm{m}}^{t}\left(f^{t}\right)$ is convex with respect to $f^{t}$ as both $-\tilde{B}^{t}\kappa W\left(f^{t}\right)^{2}$ and $VW\slash f^{t}$ are convex functions of $f^{t}$. By taking the first-order derivative of $J_{\rm{m}}^{t}\left(f^{t}\right)$ and setting it to zero, we obtain a unique solution $f_{0}^{t}=\left(\frac{V}{-2\tilde{B}^{t}\kappa}\right)^{\frac{1}{3}}>0$. If $f_{0}^{t}<f_{L}$, $J_{\rm{m}}^{t}\left(f^{t}\right)$ is increasing in $\left[f_{L},f_{U}\right]$, and thus $f^{t*}=f_{L}$; if $f_{0}^{t}>f_{U}$, $J_{\rm{m}}^{t}\left(f^{t}\right)$ is decreasing in $\left[f_{L},f_{U}\right]$, and thus $f^{t*}=f_{U}$; otherwise, if $f_{L}\leq f_{0}^{t} \leq f_{U}$, $J_{\rm{m}}^{t}\left(f^{t}\right)$ is decreasing in $\left[f_{L},f_{0}^{t}\right]$ and increasing in $\left(f_{0}^{t},f_{U}\right]$, and we have $f^{t*}=f_{0}^{t}$.
\end{proof}

It can be seen from Proposition \ref{optCPUlma} that the optimal CPU-cycle frequency is chosen by balancing the cost of the harvested energy and the execution cost. Interestingly, we find that a higher CPU-cycle frequency, i.e., lower execution delay, can be supported with a greater amount of available harvested energy, which is because that the cost of renewable energy is reduced and more energy can be used to enhance the user's QoE, as demonstrated in Corollary \ref{property1}.
\begin{corol}
The optimal CPU-cycle frequency for local execution $f^{t*}$ is independent with the channel gain $h^{t}$, and non-decreasing with the virtual energy queue length $\tilde{B}^{t}$.
\label{property1}
\end{corol}
\begin{proof}
Since $\mathcal{P}_{\rm{ME}}$ does not depend on $h^{t}$, the optimal CPU-cycle frequency is independent with the channel state. As $f_{L}$ and $f_{U}$ are constants independent with $\tilde{B}^{t}$, and $f_{0}^{t}$ increases with $\tilde{B}^{t}$ for $\tilde{B}^{t}<0$, we can conclude that $f^{t*}$ is non-decreasing with $\tilde{B}^{t}$ based on (\ref{optCPU}).
\end{proof}

Next, we will consider the case that the task is executed by the MEC server, where the optimal transmit power for computation offloading can be obtained by solving the following optimization problem $\mathcal{P}_{\rm{SE}}$:
\begin{align}
&\mathcal{P}_{\rm{SE}}: \min_{p^{t}} -\tilde{B}^{t}\cdot \frac{p^{t}L}{r\left(h^{t},p^{t}\right)} +V\cdot \frac{L}{r\left(h^{t},p^{t}\right)}\nonumber\\
&\ \ \ \ \ \ \ {\rm{s.t.}}\ \ 0<p^{t}\leq p^{\max}_{\rm{tx}}\label{maxtranspwr2}\\
&\ \ \ \ \ \ \ \ \ \ \ \ \ \frac{L}{r\left(h^{t},p^{t}\right)}\leq \tau_{d} \label{delayCE}\\
&\ \ \ \ \ \ \ \ \ \ \ \ \ \frac{p^{t}L}{r\left(h^{t},p^{t}\right)} \in\left[E_{\min},E_{\max}\right] \label{energyCE},
\end{align}
which is obtained by plugging $I_{\rm{s}}^{t}=1$, $I_{\rm{m}}^{t}=I_{\rm{d}}^{t}=0$ and $f^{t}=0$ into $\mathcal{P}_{\rm{CO}}$, and using the fact that $p^{t}>0$ for computation offloading. (\ref{delayCE}) and (\ref{energyCE}) stand for the execution delay constraint and the battery output energy constraint for MEC, respectively. We denote the objective function of $\mathcal{P}_{\rm{SE}}$ as $J_{\rm{s}}^{t}\left(p^{t}\right)$. Due to the wireless fading, it may happen that computation offloading is infeasible. In order to derive the feasibility condition and the optimal solution for $\mathcal{P}_{\rm{SE}}$ given it is feasible, we first provide the following lemma to facilitate the analysis.
\begin{lma}
For $h > 0$, $g_{1}\left(h,p\right)\triangleq \frac{p}{r\left(h,p\right)}$ is an increasing function of $p$ ($p>0$) that takes value from $\left(\sigma \ln2 \left(\omega h\right)^{-1},+\infty\right)$.
\label{nondecreasingEng}
\end{lma}
\begin{proof}
The proof is omitted due to space limitation.
\end{proof}

%\begin{proof}
%The first-order partial derivative of $g_{1}\left(h,p\right)$ with respect to $p$ is given by
%\begin{equation}
%\frac{\partial g_{1}\left(h,p\right)}{\partial p}=\frac{ %\log_{2}\left(1+\frac{hp}{\sigma}\right)-\frac{ \frac{h p}{\sigma} }{\left(1 + %\frac{hp}{\sigma}\right)\ln 2}}{\omega \log_{2}^{2}\left(1+\frac{hp}{\sigma}\right)},
%\end{equation}
%which is positive due to the inequality $\ln\left(1+x\right)>\frac{x}{1+x},\forall x> 0$. %Thus, $g_{1}\left(h,p\right)$ increases with $p$ ($p>0$) for a given $h>0$. In addition, since %$\lim\limits_{p\rightarrow 0} g_{1}\left(h,p\right)=\sigma \ln2 \left(\omega h\right)^{-1}$ %and $\lim\limits_{p\rightarrow +\infty} g_{1}\left(h,p\right)=+\infty$, we have %$g_{1}\left(h,p\right)\in \left(\sigma \ln2 \left(\omega h\right)^{-1},+\infty\right)$.
%\end{proof}

Based on Lemma \ref{nondecreasingEng}, we combine constraints (\ref{maxtranspwr2})-(\ref{energyCE}) into an inequality and obtain the feasibility condition for $\mathcal{P}_{\rm{SE}}$, as demonstrated in the following lemma.
\begin{lma}
$\mathcal{P}_{\rm{SE}}$ is feasible if and only if $p_{L}^{t}\leq p_{U}^{t}$, where $p_{L}^{t}$ and $p_{U}^{t}$ are defined as
\begin{equation}
p^{t}_{L}\triangleq
\begin{cases}
p_{L,\tau_{d}}^{t}, & \frac{\sigma L \ln 2}{\omega h^{t}} \geq E_{\min}\\
\max\{p_{L,\tau_{d}}^{t},p_{E_{\min}}^{t}\}, & \frac{\sigma L \ln 2}{\omega h^{t}} < E_{\min}
\end{cases}\ {\rm{and}}\
%\label{eqOLEng1}
%\end{equation}
%\begin{equation}
p^{t}_{U}\triangleq
\begin{cases}
\min\{p_{\rm{tx}}^{\max},p^{t}_{E_{\max}}\}, & \frac{\sigma L \ln 2}{\omega h^{t}} < E_{\max}\\
0, & \frac{\sigma L \ln 2}{\omega h^{t}} \geq E_{\max},
\end{cases}
\label{eqOLEng2}
\end{equation}
respectively. In (\ref{eqOLEng2}), $p_{L,\tau_{d}}^{t}\triangleq\left(2^{\frac{L}{\omega \tau_{d}}}-1\right)\sigma \slash h^{t}$,  $p_{E_{\min}}^{t}$ is the unique solution for $pL=r\left(h^{t},p\right)E_{\min}$ given $\sigma L \ln 2\left(\omega h^{t}\right)^{-1} < E_{\min}$, and $p^{t}_{E_{\max}}$ is the unique solution for $pL=r\left(h^{t},p\right)E_{\max}$ given $\sigma L \ln 2\left(\omega h^{t}\right)^{-1} < E_{\max}$.
\label{feasOLlma}
\end{lma}
\begin{proof}
The proof can be obtained based on Lemma \ref{nondecreasingEng}, which is omitted for brevity.
\end{proof}

We now develop the optimal solution for $\mathcal{P}_{\rm{SE}}$ as specified in the following proposition.
\begin{prop}
If $\mathcal{P}_{\rm{SE}}$ is feasible, i.e., $p_{L}^{t}\leq p_{U}^{t}$, its optimal solution is given by
\begin{equation}
p^{t*}=
\begin{cases}
p^{t}_{U}, &\tilde{B}^{t}\geq 0\ {\rm{or}}\ \tilde{B}^{t}<0, p_{U}^{t}<p_{0}^{t}\\
p^{t}_{L}, &\tilde{B}^{t}<0, p_{L}^{t}>p_{0}^{t}\\
p^{t}_{0}, &\tilde{B}^{t}<0, p_{L}^{t}\leq p_{0}^{t}\leq p_{U}^{t},
\end{cases}
\label{optOLpwr}
\end{equation}
where $p^{t}_{0}$ is the unique solution for equation $\Xi\left(h^{t},p,\tilde{B}^{t}\right)=0$ and $\Xi\left(h,p,\tilde{B}\right)\triangleq -\tilde{B}\log_{2}\left(1+\frac{hp}{\sigma}\right)-\frac{h}{\left(\sigma + hp\right)\ln 2}\left(V-\tilde{B}p\right)$.
\label{optOLpwrlma}
\end{prop}
\begin{proof}
When $\tilde{B}^{t}\geq 0$, since both terms in $J_{\rm{s}}^{t}\left(p^{t}\right)$ are non-increasing with $p^{t}$, we have $p^{t*}=p_{U}^{t}$. When $\tilde{B}^{t}<0$, we define $g_{2}\left(h,p,\tilde{B}\right)\triangleq -\frac{\tilde{B} p}{r\left(h,p\right)} +\frac{V}{r\left(h,p\right)}$, and thus
\begin{equation}
\frac{d g_{2}\left(h^{t},p,\tilde{B}^{t}\right)}{d p}=\frac{-\tilde{B}^{t}\log_{2}\left(1+\frac{h^{t}p}{\sigma}\right)-\frac{h^{t}}{\left(h^{t}p+\sigma\right)\ln 2}\left(-\tilde{B}^{t}p+V\right)}{\omega \log_{2}^{2}\left(1+\frac{h^{t}p}{\sigma}\right)}
\triangleq \frac{\Xi\left(h^{t},p,\tilde{B}^{t}\right)}{\omega \log_{2}^{2}\left(1+\frac{h^{t}p}{\sigma}\right)}.
\label{derig2}
\end{equation}
Since $\frac{d\Xi\left(h^{t},p,\tilde{B}^{t}\right)}{dp}>0$,
%\begin{equation}
%\begin{split}
%\frac{d\Xi\left(h^{t},p,\tilde{B}^{t}\right)}{dp}
%=\frac{h^{t}}{\ln 2}\cdot %\left[\frac{-\tilde{B}^{t}}{h^{t}p+\sigma}-\frac{-\tilde{B}^{t}\left(h^{t}p+\sigma\right)-h^{t}\left(-\tilde{B}^{t}p+V\right)}{\left(h^{t}p+\sigma\right)^{2}}\right]
%=\frac{\left(h^{t}\right)^{2}}{\ln 2}\cdot %\frac{-\tilde{B}^{t}p+V}{\left(h^{t}p+\sigma\right)^{2}}> 0,
%\end{split}
%\end{equation}
$\Xi\left(h^{t},p,\tilde{B}^{t}\right)$ increases with $p$. In addition, as $\Xi\left(h^{t},0,\tilde{B}^{t}\right)=-\frac{h^{t}V}{\sigma \ln 2}<0$ and $\lim\limits_{p\rightarrow +\infty}\Xi\left(h^{t},p,\tilde{B}^{t}\right)=+\infty$, there exists a unique $p_{0}^{t}\in\left(0,+\infty\right)$ satisfying $\Xi\left(h^{t},p_{0}^{t},\tilde{B}^{t}\right)=0, \forall h^{t}>0$. Since the denominator of (\ref{derig2}) is positive for $h^{t}>0$ and $p>0$, $\frac{dg_{2}\left(h^{t},p,\tilde{B}^{t}\right)}{dp}< 0$ for $p\in \left(0,p_{0}^{t}\right)$, i.e., $g_{2}\left(h^{t},p,\tilde{B}^{t}\right)$ is decreasing, and $\frac{dg_{2}\left(h^{t},p,\tilde{B}^{t}\right)}{dp} \geq 0$ for $p\in\left[p_{0}^{t},+\infty\right)$, i.e., $g_{2}\left(h^{t},p,\tilde{B}^{t}\right)$ is increasing. Consequently, when $\tilde{B}^{t}<0$ and $p_{L}^{t}\leq p_{0}^{t} \leq p_{U}^{t}$, $J_{\rm{s}}^{t}\left(p^{t}\right)$ is non-increasing in $\left[p_{L}^{t},p_{0}^{t}\right)$ while non-decreasing in $\left(p_{0}^{t},p_{U}^{t}\right]$, and thus $p^{t*}=p_{0}^{t}$; when $\tilde{B}^{t}<0$ and $p^{t}_{L}>p_{0}^{t}$, $J_{\rm{s}}^{t}\left(p^{t}\right)$ is non-decreasing in the feasible domain, and thus $p^{t*}=p_{L}^{t}$; otherwise when $\tilde{B}^{t}<0$ and $p^{t}_{U}<p_{0}^{t}$, $J_{\rm{s}}^{t}\left(p^{t}\right)$ is non-increasing in the feasible domain, we have $p^{t*}=p_{U}^{t}$.
\end{proof}

Similar to mobile execution, we find a monotonic behavior of the optimal transmit power for computation offloading, as shown in the following corollary.
\begin{corol}
For a given $h^{t}$ such that $\mathcal{P}_{\rm{SE}}$ is feasible, the optimal transmit power for computation offloading $p^{t*}$ is non-decreasing with $\tilde{B}^{t}$.
\label{property2}
\end{corol}
\begin{proof}
Please refer to Appendix B.
\end{proof}

\begin{rmk}
We can see from (\ref{optOLpwr}) that the optimal transmit power for computation offloading depends on both the battery energy level and the channel state. In Corollary \ref{property2}, we show a higher battery energy level awakes a higher transmit power, and thus incurs smaller execution latency. However, the monotonicity of $p^{t*}$ with respect to $h^{t}$ does not hold. This is due to the battery output energy constraint, which makes the feasible set of $p^{t}$ change with $h^{t}$.
\end{rmk}

Based on Proposition \ref{optCPUlma} and \ref{optOLpwrlma}, the optimal computation offloading decision can be obtained by evaluating the optimal values of $\mathcal{P}_{\rm{CO}}$ for the three computation modes, i.e., dropping the task, mobile execution and MEC server execution, which can be explicitly expressed as
\begin{equation}
\langle \bm{I}^{t*},f^{t*},p^{t*}\rangle = \arg\min_{\langle\bm{I}^{t},f^{t},p^{t}\rangle \in \mathcal{F}^{t}_{\rm{CO}} } J_{\rm{CO}}\left(\bm{I}^{t},f^{t},p^{t}\right),
\end{equation}
where $J_{\rm{CO}}\left(\bm{I}^{t},f^{t},p^{t}\right)=\bm{1}_{I_{\rm{m}}^{t}=1}J^{t}_{\rm{m}}\left(f^{t}\right)+\bm{1}_{I_{\rm{s}}^{t}=1}J^{t}_{\rm{s}}\left(p^{t}\right)+\bm{1}_{I_{\rm{d}}^{t}=1,\zeta^{t}=1}\cdot V\phi$, and $V\phi$ is the value of $J_{\rm{CO}}\left(\bm{I}^{t},f^{t},p^{t}\right)$ when a computation task is dropped. {Note that when $\zeta^{t}=1$ and $\mathcal{F}_{\rm{CO}}^{t}=\{\langle \left[I_{\rm{m}}^{t}=0,I_{\rm{s}}^{t}=0,I_{\rm{d}}^{t}=1\right],0,0\rangle\}$, the computation task has to be dropped, as $\mathcal{P}_{\rm{CO}}$ has only one feasible solution.} It is also worth mentioning that bisection search can be applied to obtain $p^{t}_{L}$, $p^{t}_{U}$ and $p^{t}_{0}$, i.e., solving $\mathcal{P}_{\rm{CO}}$ is of low complexity.

\section{Performance Analysis}
In this section, we will first prove the feasibility of the LODCO algorithm for $\mathcal{P}_{2}$, and the achievable performance of the proposed algorithm will then be analyzed.

%In this section, the performance of the LODCO algorithm will be analyzed. We will first prove the feasibility of the LODCO algorithm for $\mathcal{P}_{2}$. During the analysis, an auxiliary optimization problem $\mathcal{P}_{3}$ will be introduced, which bridges the optimal performance of $\mathcal{P}_2$ and the performance achieved by the proposed algorithm. Together with Proposition \ref{tightperf}, this will establish the asymptotic optimality of the LODCO algorithm for $\mathcal{P}_{1}$.

\subsection{Feasibility}
We verify the feasibility of the LODCO algorithm by showing that under the optimal solution for the per-time slot problem, the energy causality constraint in (\ref{EHcausality}) is always satisfied, as demonstrated in the following proposition.
\begin{prop}
Under the optimal solution for the per-time slot problem, when $B^{t}<\tilde{E}_{\max}$, $I_{\rm{d}}^{t}=1$, $I_{\rm{m}}^{t}=I_{\rm{s}}^{t}=0$, $f^{t}=0$, and $p^{t}=0$, and the energy causality constraint in (\ref{EHcausality}) will not be violated, i.e., the LODCO algorithm is feasible for $\mathcal{P}_{2}$ ($\mathcal{P}_{1}$).
\label{threszero}
\end{prop}
\begin{proof}
When $B^{t}<\tilde{E}_{\max}$, we will show by contradiction that with the optimal computation offloading decision, $\mathcal{E}\left(\bm{I}^{t},f^{t},p^{t}\right)=0$. Suppose there exists an optimal computation offloading decision $\langle \bm{I}^{t*},f^{t*},p^{t*}\rangle$ with either $I_{\rm{m}}^{t*}=1$ or $I_{\rm{s}}^{t*}=1$. With this solution, due to the non-zero lower bound of the battery output energy, i.e., (\ref{tightenedconst}), the value of $J_{\rm{CO}}\left(\bm{I}^{t*},f^{t*},p^{t*}\right)$ will be no less than $-\tilde{B}^{t}E_{\min}$, which is greater than $V\phi$ as achieved by the solution with $I_{\rm{d}}^{t}=1$, i.e., $\langle \bm{I}^{t*},f^{t*},p^{t*}\rangle$ is not optimal for the per-time slot problem. When $B^{t}\geq \tilde{E}_{\max}$, as $\max\limits_{\langle\bm{I}^{t},f^{t},p^{t}\rangle\in\mathcal{F}^{t}_{\rm{CO}}}\mathcal{E}\left(\bm{I}^{t},f^{t},p^{t}\right)=\tilde{E}_{\max}$, $\mathcal{E}\left(\bm{I}^{t},f^{t},p^{t}\right)\leq B^{t},\forall \langle \bm{I}^{t},f^{t},p^{t}\rangle \in \mathcal{F}^{t}_{\rm{CO}}$.  Thus, (\ref{EHcausality}) holds under the LODCO algorithm.
\end{proof}

Based on the optimal energy harvesting decision and Proposition \ref{threszero}, we show the battery energy level is confined within an interval as shown in the following corollary.
\begin{corol}
Under the LODCO algorithm, the battery energy level at the mobile device $B^{t}$ is confined within $\left[0,\theta + E_{H}^{\max}\right],\forall t\in\mathcal{T}$.
\label{batterybound}
\end{corol}
\begin{proof}
The lower bound of $B^{t}$ is straightforward as the energy causality constraint is not violated according to Proposition \ref{threszero}. The upper bound of $B^{t}$ can be obtained based on the optimal energy harvesting in (\ref{optEH}): Suppose $\theta <B^{t} \leq \theta + E_{H}^{\max}$, since $e^{t*}=0$, we have $B^{t+1}\leq B^{t}\leq \theta+E_{H}^{\max}$; otherwise, if $B^{t}\leq \theta$, since $e^{t*}=E_{H}^{t}$, we have $B^{t+1}\leq B^{t}+e^{t*}\leq \theta+e^{t*}\leq \theta +E_{H}^{\max}$. Consequently, we have $B^{t}\in\left[0,\theta +E_{H}^{\max}\right],\forall t\in\mathcal{T}$.
\end{proof}

As will be seen in the next subsection, the bounds of the battery energy level are useful for deriving the main result on the performance of the proposed algorithm. In addition, Corollary \ref{batterybound} indicates that, given the size of the available energy storage $C_{B}$, we can determine the control parameter $V$ as
%\begin{equation}
$\phi^{-1}\cdot \left(C_{B}-E_{H}^{\max}-\tilde{E}_{\max}\right)E_{\min}$, where $C_{B}$ should be greater than $\tilde{E}_{\max}+E_{H}^{\max}$ in order to guarantee $V>0$. This is instructive for installation of EH and storage units at the mobile devices.
%\end{equation}

\subsection{Asymptotic Optimality}
In this subsection, we will analyze the performance of the LODOC algorithm, where an auxiliary optimization problem $\mathcal{P}_{3}$ will be introduced to bridge the optimal performance of $\mathcal{P}_2$ and the performance achieved by the proposed algorithm. This will demonstrate the asymptotic optimality of the LODCO algorithm for $\mathcal{P}_{1}$ conjointly with Proposition \ref{tightperf}.

Firstly, we define the Lyapunov function as
\begin{equation}
L\left(\tilde{B}^{t}\right)=\frac{1}{2}\left(\tilde{B}^{t}\right)^{2}=\frac{1}{2}\left(B^{t}-\theta\right)^{2}.
\label{lyafunc}
\end{equation}
Accordingly, the Lyapunov drift function and the Lyapunov drift-plus-penalty function can be expressed as
\begin{equation}
\Delta\left(\tilde{B}^{t}\right)=\mathbb{E}\left[L\left(\tilde{B}^{t+1}\right)-L\left(\tilde{B}^{t}\right)|\tilde{B}^{t}\right]
\label{lyadrift}
\end{equation}
and
\begin{equation}
\Delta_{V}\left(\tilde{B}^{t}\right)=\Delta\left(\tilde{B}^{t}\right)+V\mathbb{E}\left[\mathcal{D}\left(\bm{I}^{t},f^{t},p^{t}\right)+\phi \cdot \bm{1}\left(\zeta^{t}=1,I_{\rm{d}}^{t}=1\right)|\tilde{B}^{t}\right],
\end{equation}
respectively.

In the following lemma, we derive an upper bound for $\Delta_{V}\left(\tilde{B}^{t}\right)$, which will play an important part throughout the analysis of the LODCO algorithm.

\begin{lma}
For arbitrary feasible decision variables $e^{t}$, $\bm{I}^{t}$, $f^{t}$ and $p^{t}$ for $\mathcal{P}_{2}$, $\Delta_{V}\left(\tilde{B}^{t}\right)$ is upper bounded by
\begin{equation}
%\begin{split}
\Delta_{V}\left(\tilde{B}^{t}\right)\leq \mathbb{E}\bigg[B^{t}\left[e^{t}-\mathcal{E}\left(\bm{I}^{t},f^{t},p^{t}\right)\right]
+V\left[\mathcal{D}\left(\bm{I}^{t},f^{t},p^{t}\right)+\phi \cdot \bm{1}\left(\zeta^{t}=1,I_{\rm{d}}^{t}=1\right)\right]|\tilde{B}^{t}\bigg]+C,
%\end{split}
\end{equation}
where $C=\frac{\left(E_{H}^{\max}\right)^{2}+\left(\tilde{E}_{\max}\right)^{2}}{2}$.
\label{lmaupperboundlyaV}
\end{lma}
\begin{proof}
Please refer to Appendix C.
\end{proof}

Note that the terms inside the conditional expectation of the upper bound derived in Lemma \ref{lmaupperboundlyaV} coincides with the objective function of the per-time slot problem in the LODCO algorithm. To facilitate the performance analysis, we define the following auxiliary problem $\mathcal{P}_{3}$:
\begin{align}
&\mathcal{P}_{3}:\ \min\limits_{\bm{I}^{t},f^{t},p^{t},e^{t}}\lim\limits_{T\rightarrow \infty}\frac{1}{T}\mathbb{E}\left[\sum_{t=0}^{T-1}{\rm{cost}}^{t}\right]\nonumber\\
&\ \ \ \ \ \ \ \ \ {\rm{s.t.}}\ \ (\ref{indicatorconstraint}), (\ref{harvestableconst}), (\ref{deadconst})-(\ref{tightenedconst})\nonumber\\
&\ \ \ \ \ \ \ \ \ \ \ \ \ \lim_{T\rightarrow +\infty}\frac{1}{T}\sum_{t=0}^{T-1}\mathbb{E}\left[\mathcal{E}\left(\bm{I}^{t},f^{t},p^{t}\right)-e^{t}\right]=0\label{relaxEHconstraint}.
\end{align}
In $\mathcal{P}_{3}$, the average harvested energy consumption equals the average harvested energy, i.e., the energy causality constraint in $\mathcal{P}_{2}$ is replaced by (\ref{relaxEHconstraint}). Denote the optimal value of $\mathcal{P}_{3}$ as ${\rm{EC}}_{\mathcal{P}_{3}}^{*}$. In the following lemma, we will show that $\mathcal{P}_{3}$ is a relaxation of $\mathcal{P}_{2}$.

\begin{lma}
$\mathcal{P}_{3}$ is a relaxation of $\mathcal{P}_{2}$, i.e., ${\rm{EC}}_{\mathcal{P}_{3}}^{*}\leq {\rm{EC}}_{\mathcal{P}_{2}}^{*}$.
\label{relaxProb}
\end{lma}
\begin{proof}
The proof can be obtained by showing any feasible solution for $\mathcal{P}_{2}$ is also feasible for $\mathcal{P}_{3}$, which is omitted for brevity.
\end{proof}

%\begin{proof}
%For any feasible solution for $\mathcal{P}_{2}$, based on the battery dynamics, we have
%\begin{equation}
%B^{t+1}=B^{t}-\mathcal{E}\left(\bm{I}^{t},f^{t},p^{t}\right)+e^{t},t=0,\cdots,T-1.
%\end{equation}
%Summing up both sides of the above $T$ equalities, taking the expectation, dividing both sides %by $T$ and letting $T$ goes to infinity, we have
%\begin{equation}
%\lim_{T\rightarrow +\infty} \frac{1}{T}\mathbb{E}\left[B^{T}\right]=\lim_{T\rightarrow %+\infty}\frac{1}{T}\mathbb{E}\left[B^{0}\right]-\lim_{T\rightarrow %+\infty}\frac{1}{T}\sum_{t=0}^{T-1}\mathbb{E}\left[\mathcal{E}\left(\bm{I}^{t},f^{t},p^{t}\right)-e^{t}\right].
%\end{equation}
%Since $B^{t}<+\infty$, we have $\lim\limits_{T\rightarrow %+\infty}\frac{1}{T}\mathbb{E}\left[B^{0}\right]=\lim\limits_{T\rightarrow %+\infty}\frac{1}{T}\mathbb{E}\left[B^{T}\right]=0$, i.e., (\ref{relaxEHconstraint}) is %satisfied. Hence, any feasible solution for $\mathcal{P}_{2}$ is also feasible for %$\mathcal{P}_{3}$, which ends the proof.
%\end{proof}

Besides, in the following lemma, we show the existence of a stationary and randomized policy \cite{Neely10}, where the decisions are i.i.d. among different time slots and depend only on $E_{H}^{t}$, $\zeta^{t}$ and $h^{t}$, that behaves arbitrarily close to the optimal solution of $\mathcal{P}_{3}$, meanwhile, the difference between $\mathbb{E}\left[e^{t}\right]$ and $\mathbb{E}\left[\mathcal{E}\left(\bm{I}^{t},f^{t},p^{t}\right)\right]$ is arbitrarily small.

\begin{lma}
For an arbitrary $\delta >0 $, there exists a stationary and randomized policy $\Pi$ for $\mathcal{P}_{3}$, which decides $e^{t\Pi}$, $\bm{I}^{t\Pi}$, $f^{t\Pi}$ and $p^{t\Pi}$, such that (\ref{indicatorconstraint}), (\ref{harvestableconst}), (\ref{deadconst})-(\ref{tightenedconst}) are met, and the following inequalities are satisfied:
\begin{equation}
\mathbb{E}\left[\mathcal{D}\left(\bm{I}^{t\Pi},f^{t\Pi},p^{t\Pi}\right)+\phi\cdot \bm{1}\left(\zeta^{t}=1, I_{\rm{d}}^{t\Pi}\right)\right]\leq {\rm{EC}}^{*}_{\mathcal{P}_{3}}+\delta, t\in\mathcal{T},
\end{equation}
\begin{equation}
\bigg|\mathbb{E}\left[\mathcal{E}\left(\bm{I}^{t\Pi},f^{t\Pi},p^{t\Pi}\right)-e^{t\Pi}\right]\bigg|\leq \varrho \delta, t\in\mathcal{T},
\end{equation}
where $\varrho$ is a scaling constant.
\label{lmaarbclose}
\end{lma}
\begin{proof}
The proof can be obtained by Theorem 4.5 in \cite{Neely10}, which is omitted for brevity.
\end{proof}

In Section IV, we bounded the optimal performance of the modified ECM problem $\mathcal{P}_{2}$ with that of the original ECM problem $\mathcal{P}_{1}$, while in Lemma \ref{relaxProb}, we showed the auxiliary problem $\mathcal{P}_{3}$ is a relaxation of $\mathcal{P}_{2}$. With the assistance of these results, next, we will provide the main result in this subsection, which characterizes the worst-case performance of the LODCO algorithm.

\begin{thm}
The execution cost achieved by the proposed LODCO algorithm, denoted as ${\rm{EC}}_{\rm{LODCO}}$, is upper bounded by
\label{thmasymopt}
\end{thm}
\begin{equation}
{\rm{EC}}_{\rm{LODCO}}\leq {\rm{EC}}_{\mathcal{P}_{1}}^{*}+\nu\left(E_{\min}\right) + C\cdot V^{-1}.
\end{equation}
\begin{proof}
Please refer to Appendix D.
\end{proof}

\begin{rmk}
Theorem \ref{thmasymopt} indicates that the execution cost upper bound can be made arbitrarily tight by letting $V\rightarrow +\infty$, $E_{\min}\rightarrow 0$, that is, the proposed algorithm asymptotically achieves the optimal performance of the original design problem $\mathcal{P}_{1}$. However, the optimal performance of $\mathcal{P}_{1}$ is achieved at the price of a higher battery capacity requirement and longer convergence time to the optimal performance. This is because that, the battery energy level will be stabilized around $\theta$ under the LODCO algorithm. As $E_{\min}$ decreases or $V$ increases, $\theta$ increases accordingly, and it will need a longer time to accumulate the harvested energy, which postpones the arrival of the system stability and hence delays the convergence. Thus, by adjusting the control parameters, we can balance the system performance and the battery capacity/convergence time. Similar phenomenon was observed in our previous work \cite{YMao1512}.
\end{rmk}

\section{Simulation Results}
In this section, we will verify the theoretical results derived in Section V and evaluate the performance of the proposed LODCO algorithm through simulations. In simulations, $E_{H}^{t}$ is uniformly distributed between 0 and $E_{H}^{\max}$ with the average EH power given by $P_{H}=E_{H}^{\max}\left(2\tau\right)^{-1}$, and the channel power gains are exponentially distributed with mean $g_{0}d^{-4}$, where $g_{0}=-40$ dB is the path-loss constant. In addition, $\kappa =10^{-28}$, $\tau = \phi = 2$ ms, $w=1$ MHz, $\sigma = 10^{-13}$ W, $p_{\rm{tx}}^{\max}=1$ W, $f_{\rm{CPU}}^{\max}=1.5$ GHz, $E_{\max}=2$ mJ, and $L=1000$ bits. Besides, $X=5900$ cycles per byte, which corresponds to the workload of processing the English main page of Wikipedia \cite{Miettinen10}. Moreover, $P_{H}=12$ mW, $d=50$ m and $\tau_{d}=2$ ms unless otherwise specified. For comparison, we introduce three benchmark policies, namely, \emph{mobile execution with greedy energy allocation} (Mobile Execution (GD)), \emph{MEC server execution with greedy energy allocation} (MEC Server Execution (GD)) and \emph{dynamic offloading with greedy energy allocation} (Dynamic offloading (GD)), which minimize the execution cost at the current time slot. They work as follows:
\begin{itemize}
\item \textbf{Mobile Execution (GD):} Compute the maximum feasible CPU-cycle frequency as $f_{U}^{t}=\min\{f_{\rm{CPU}}^{\max},\sqrt{\frac{\min\{B^{t},E_{\max}\}}{\kappa W}}\}$ when $\zeta^{t}=1$. If $W\slash f_{U}^{t}\leq \tau_{d}$, the computation task will be executed locally with CPU-cycle frequency $f_{U}^{t}$; otherwise, mobile execution is infeasible and the task will be dropped. Note that computation offloading is disabled in this policy.
\item \textbf{MEC Server Execution (GD):} When $\zeta^{t}=1$, compute the maximum feasible transmit power as $p^{t}_{U}=\min\{p_{\rm{tx}}^{\max},p^{t}_{\min\{B^{t},E_{\max}\}}\}$ if $\sigma L \ln 2\left(\omega h^{t}\right)^{-1}< \min\{B^{t},E_{\max}\}$, where $p^{t}_{\min\{B^{t},E_{\max}\}}$ is the unique solution of $pL=r\left(h^{t},p\right)\min\{B^{t},E_{\max}\}$. If $L\slash r\left(h^{t},p^{t}_{U}\right)\leq \tau_{d}$, the computation task will be offloaded to the MEC server with transmit power $p_{U}^{t}$; otherwise, MEC server execution is infeasible and the computation task will be dropped. Note that the computation tasks are always offloaded to the MEC server in this policy.
\item \textbf{Dynamic Offloading (GD):} When $\zeta^{t}=1$, compute $f_{U}^{t}$ and $p_{U}^{t}$ as in the Mobile Execution (GD) and MEC Server Execution (GD) policies, respectively, and check if they can meet the delay requirement. Then the feasible computation mode that incurs smaller execution delay will be chosen. If neither computation modes is feasible, the computation task will be dropped.
\end{itemize}

\subsection{Theoretical Results Verification}

\begin{figure}[h]
\begin{center}
   \includegraphics[width=0.6\textwidth]{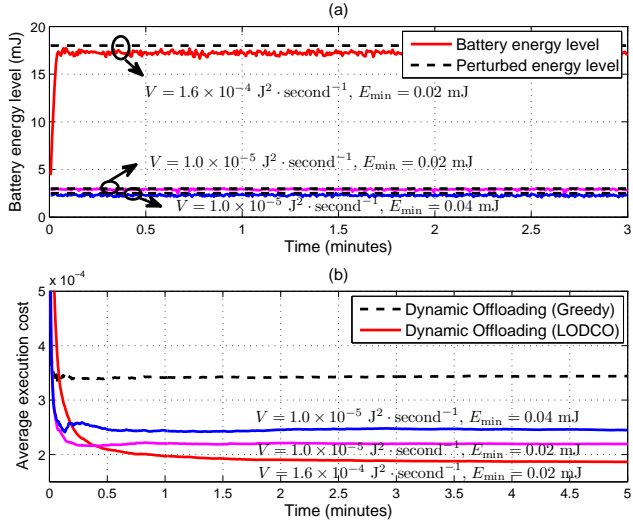}
\end{center}
\vspace{-30pt}
\caption{Battery energy level and average execution cost vs. time, $\rho=0.6$.}
\label{EnergyECvstime}
\end{figure}
In this subsection, we will verify the feasibility and asymptotic optimality of the LODCO algorithm developed in
Proposition \ref{threszero}, Corollary \ref{batterybound}, and Theorem \ref{thmasymopt}, respectively. The value of $\theta$ is chosen as the value of the right-hand side of (\ref{perturbpara}). In Fig. \ref{EnergyECvstime}(a), the battery
energy level is depicted to demonstrate the feasibility of the LODOC algorithm for $\mathcal{P}_{2}$ ($\mathcal{P}_{1}$). First, we observe that the harvested energy keeps accumulating at the beginning, and finally stabilizes
around the perturbed energy level. This is due to the fact that in the proposed algorithm the Lyapunov drift-plus-penalty function is minimized at each time slot. From the curves, with a larger value of $V$ or a smaller value of $E_{\min}$, the stabilized energy level becomes higher, which agrees with the definition of the perturbation parameter in (\ref{perturbpara}). Also, we see that the energy level is confined within $\left[0,\theta+E_{H}^{\max}\right]$, which verifies Corollary \ref{batterybound} and confirms that the energy causality constraint is not violated, i.e., Proposition \ref{threszero} holds. The evolution of the average execution cost with respect to time is
shown in Fig. \ref{EnergyECvstime}(b). We see that, a larger value of $V$ or a smaller value of $E_{\min}$ results in a smaller long-term average execution cost.
Nevertheless, the algorithm converges more slowly to the stable performance. Besides, if $\langle E_{\min},V\rangle$ are properly selected,
the proposed algorithm will achieve significant performance gain compared to the benchmark policies.

\begin{figure}[h]
\begin{center}
   \includegraphics[width=0.6\textwidth]{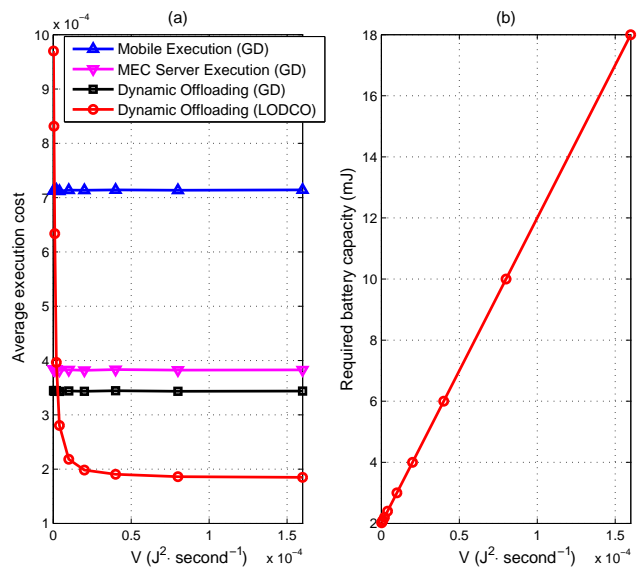}
\end{center}
\vspace{-30pt}
\caption{Average execution cost and required battery capacity vs. $V$, $\rho=0.6$ and $E_{\min}=0.02$ mJ.}
\label{ECvsV}
\end{figure}

The relationship between the average execution cost/required battery capacity and $V$ is shown in Fig. \ref{ECvsV}. We see from Fig. \ref{ECvsV}(a) that the execution cost achieved by the proposed algorithm decreases inversely proportional to $V$, and eventually it converges to the optimal value of $\mathcal{P}_{2}$, which verifies the asymptotic optimality developed in Theorem \ref{thmasymopt}. However, as shown from Fig. \ref{ECvsV}(b), the required battery capacity grows linearly with $V$ since the value of $\theta$ is linearly increasing with $V$. Thus, $V$ should be chosen to balance the achievable performance, convergence time and required battery capacity. For instance, if a battery with 18 mW capacity is available, we can choose $V=1.6\times 10^{-4}\ {\rm{J}}^{2}\cdot {\rm{second}}^{-1}$ for the LODCO algorithm, and then 74.4\%, 51.8\% and 46.3\% performance gain compared to the Mobile Execution (GD), MEC Server Execution (GD) and Dynamic Offloading (GD) policies, respectively, will be obtained.

\subsection{Performance Evaluation}
\begin{figure}[!htbp]
  \centering
  \subfigure[Execution cost vs. $\rho$]{
    \label{ECRho} %% label for first subfigure
    \raisebox{-1cm}{\includegraphics[width=0.5\textwidth]{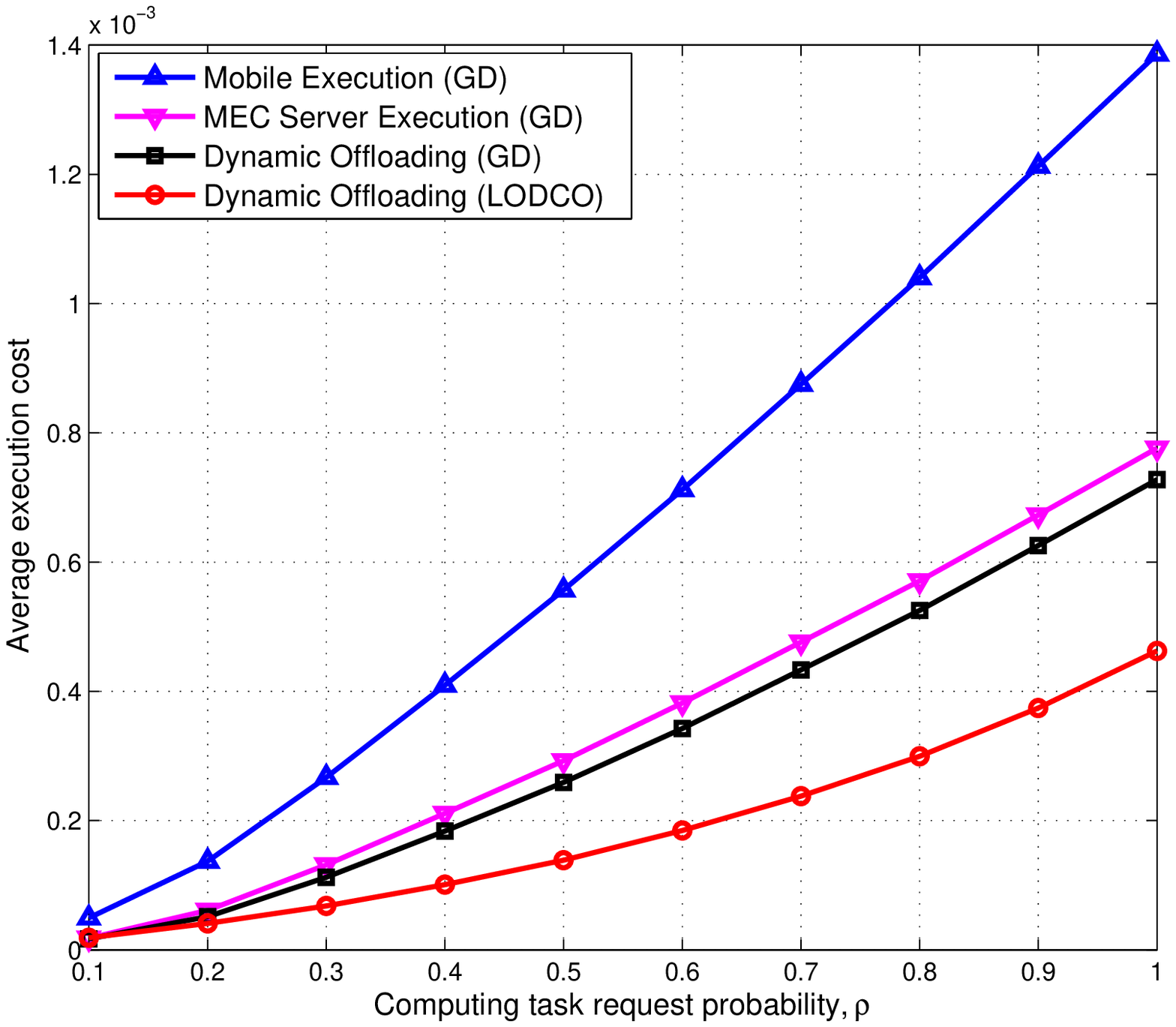}} }  %Case1.eps
  \hspace{-25pt}
  \subfigure[Average completion time/task drop ratio vs. $\rho$]{
    \label{CTDRRho} %% label for second subfigure
    \raisebox{-1cm}{\includegraphics[width=0.5\textwidth]{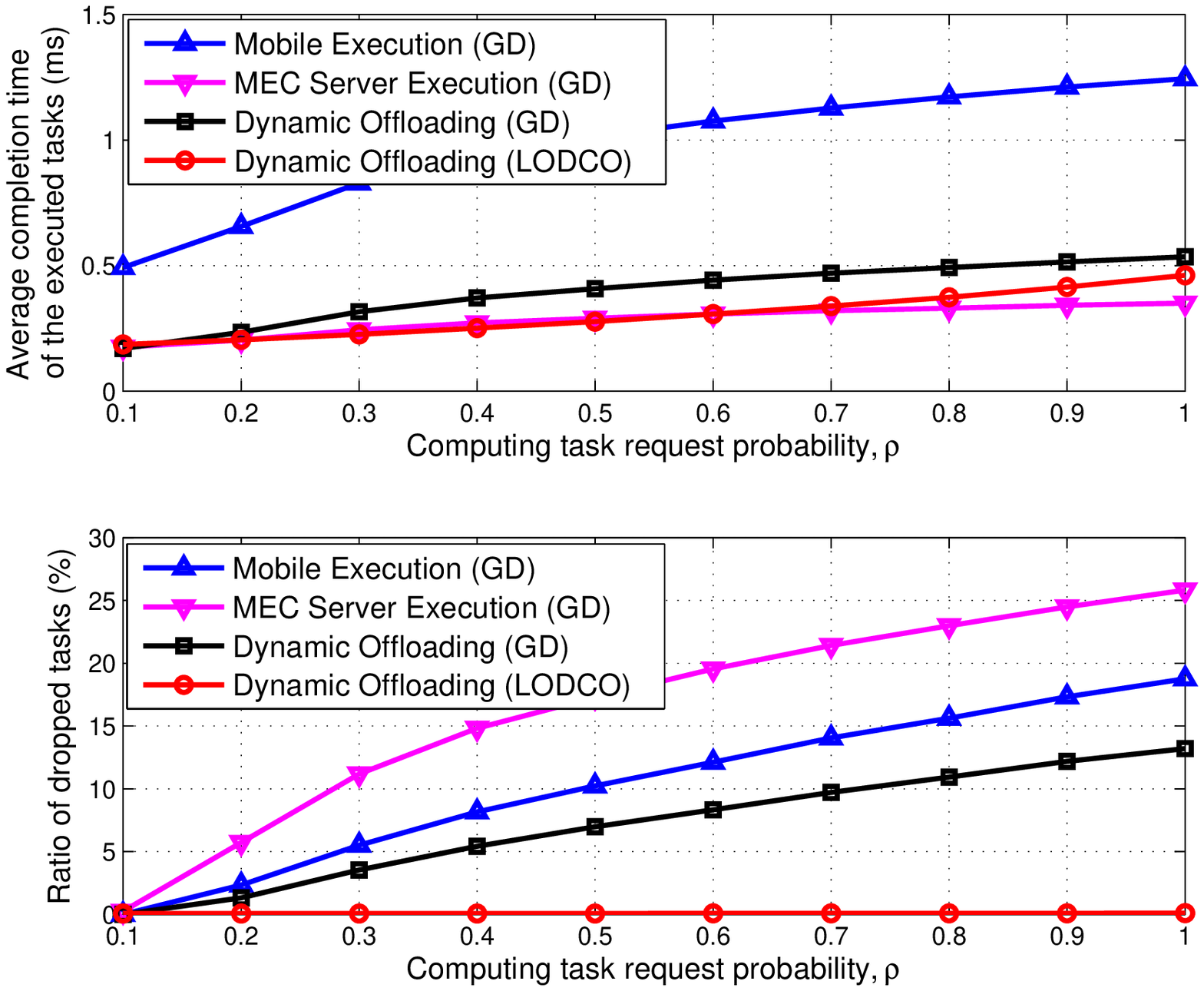}}}     %Case2.eps
  \vspace{-5pt}
  \caption{System performance vs. task arrival probability.}
  \label{ImpactRho} %% label for entire figure
\end{figure}

We will show the effectiveness of the proposed algorithm and demonstrate the impacts of various system parameters in this subsection. First, the impacts of the task request probability $\rho$ on the system performance, including the execution cost, the average completion time of the executed tasks and the task drop ratio, are illustrated in Fig. \ref{ImpactRho}. We see in Fig. \ref{ECRho} that the execution cost increases with $\rho$, which is in accordance with our intuition. Besides, the LODCO algorithm achieves significant execution cost reduction compared to the benchmark policies. In Fig. \ref{CTDRRho}, the average completion time of the executed tasks and the task drop ratio are shown, We see that the LODCO algorithm achieves a near-zero task drop ratio, while those achieved by the benchmark policies increase rapidly with $\rho$. In terms of the average completion time, the LODCO algorithm outperforms the benchmark policies when $\rho$ is small. However, when $\rho$ is large, the average completion time achieved by the LODCO algorithm is slightly longer than that achieved by the MEC Server Execution (GD) policy. The reason is, in order to minimize the execution cost, the LODCO algorithm suppresses the task drop ratio at the expense of a minor execution delay performance degradation.

\begin{figure}[!htbp]
  \centering
  \subfigure[Execution cost vs. $P_{H}$]{
    \label{ECEHrate} %% label for first subfigure
    \raisebox{-1cm}{\includegraphics[width=0.5\textwidth]{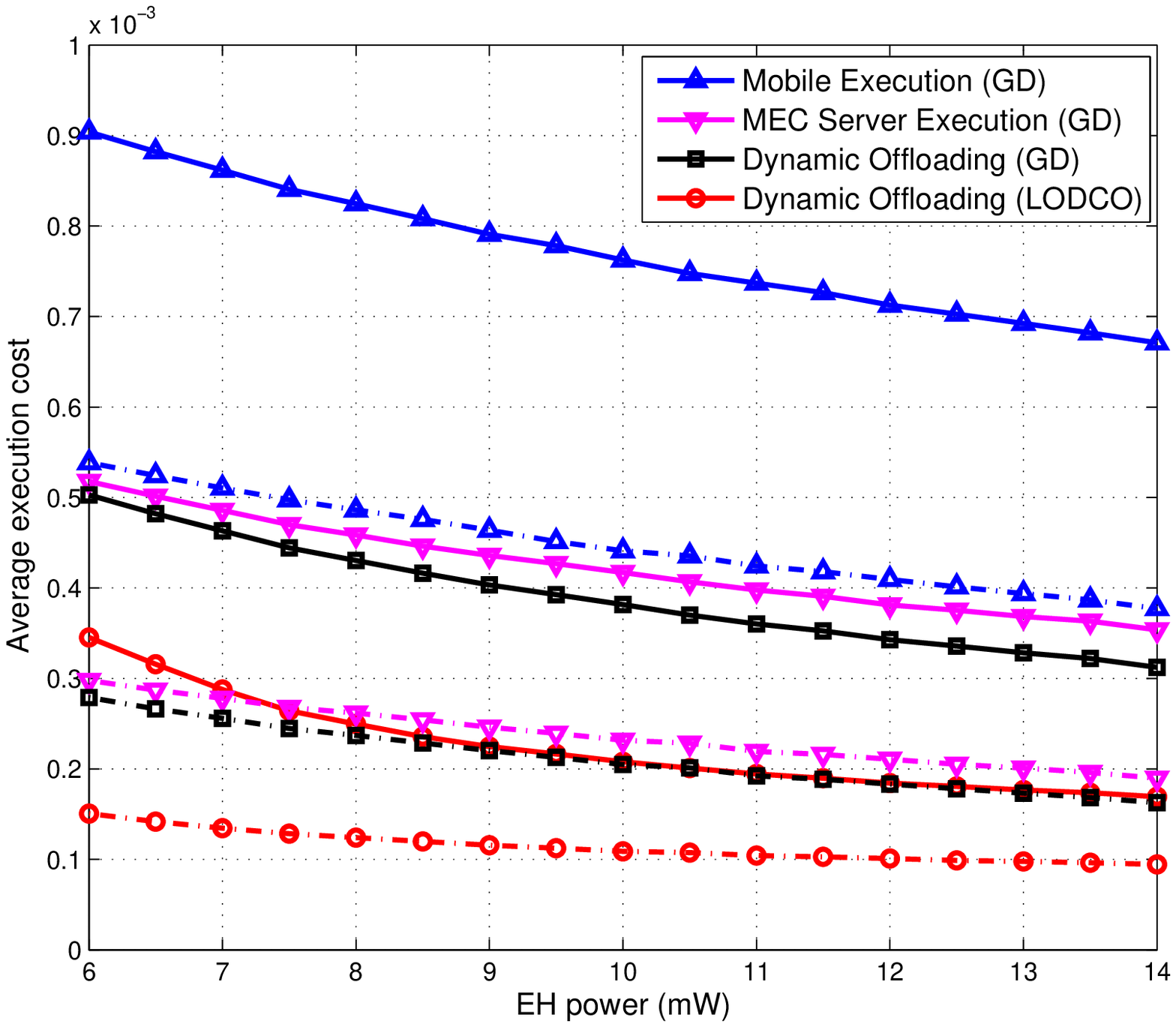}} }  %Case1.eps
  \hspace{-25pt}
  \subfigure[Average completion time/task drop ratio vs. $P_{H}$]{
    \label{CTDREHrate} %% label for second subfigure
    \raisebox{-1cm}{\includegraphics[width=0.5\textwidth]{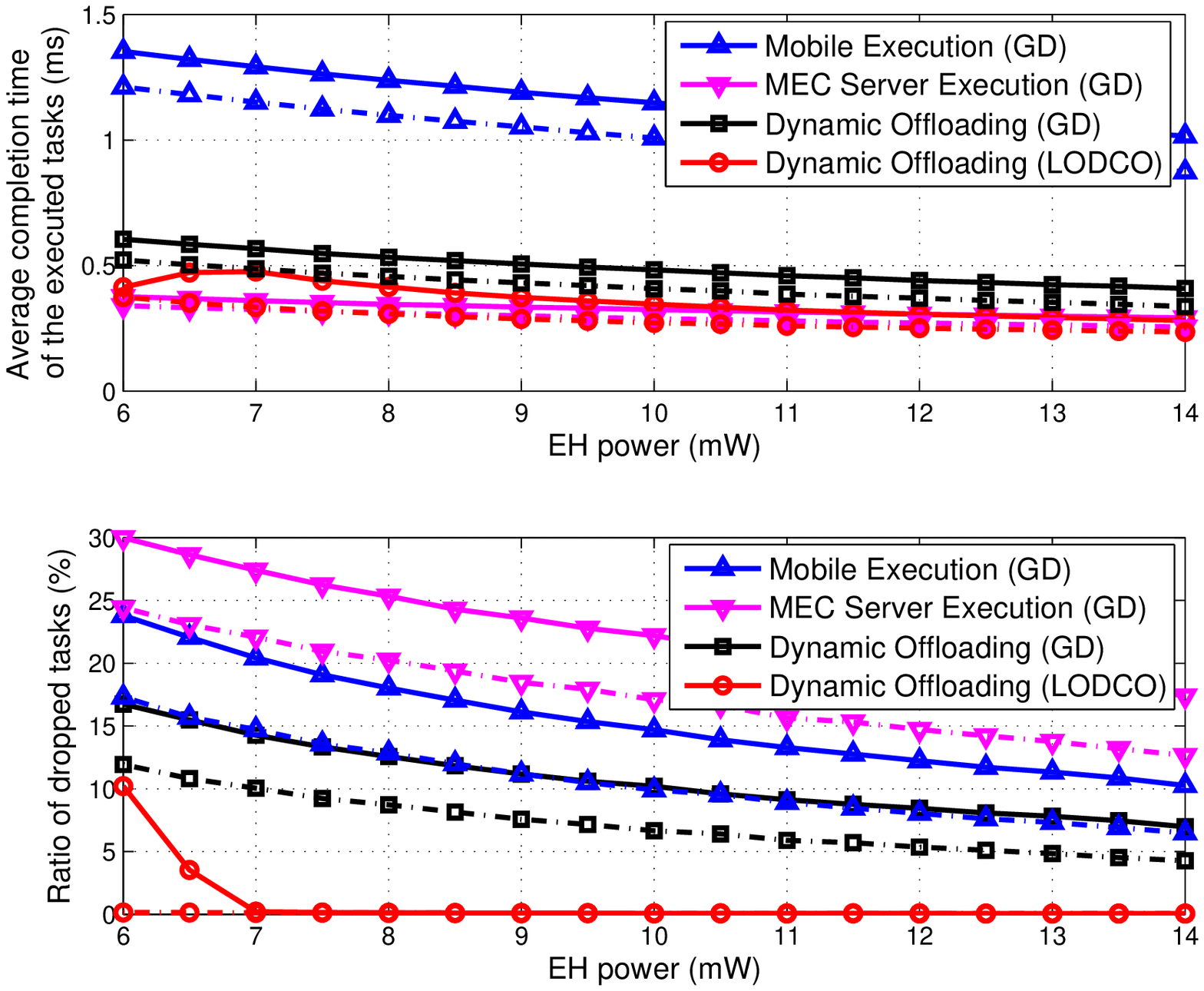}}}     %Case2.eps
  \vspace{-5pt}
  \caption{System performance vs. EH rate,  the solid lines corresponds to $\rho=0.6$ and the dash-solid lines corresponds to $\rho=0.4$.}
  \label{ImpactEHrate} %% label for entire figure
\end{figure}

The system performance versus the EH rate, i.e., $P_{H}$, is shown in Fig. \ref{ImpactEHrate}, where the effectiveness of the LODCO algorithm is again validated. In addition, we see the execution cost decreases as the EH rate increases since consuming the renewable energy incurs no cost. Similar to the execution cost, the task drop ratios achieved by different policies decrease with the EH rate. Interestingly, under the LODCO algorithm, an increase of the EH rate does not necessarily reduce the average completion time, e.g., when $\rho=0.6$ and $P_{H}$ increases from $6$ to $7$ mW, the LODCO algorithm has introduced a $0.1$ ms extra average completion time, but secured a 10\% task drop reduction. Since the optimization objective is the execution cost, eliminating task drops brings more benefits in terms of system cost when the system resource is scarce, i.e., the harvested energy is insufficient compared to the relatively intense computation workload.

\begin{figure}[!htbp]
  \centering
  \subfigure[Execution cost vs. $\tau_{d}$]{
    \label{ECDeadline} %% label for first subfigure
    \raisebox{-1cm}{\includegraphics[width=0.5\textwidth]{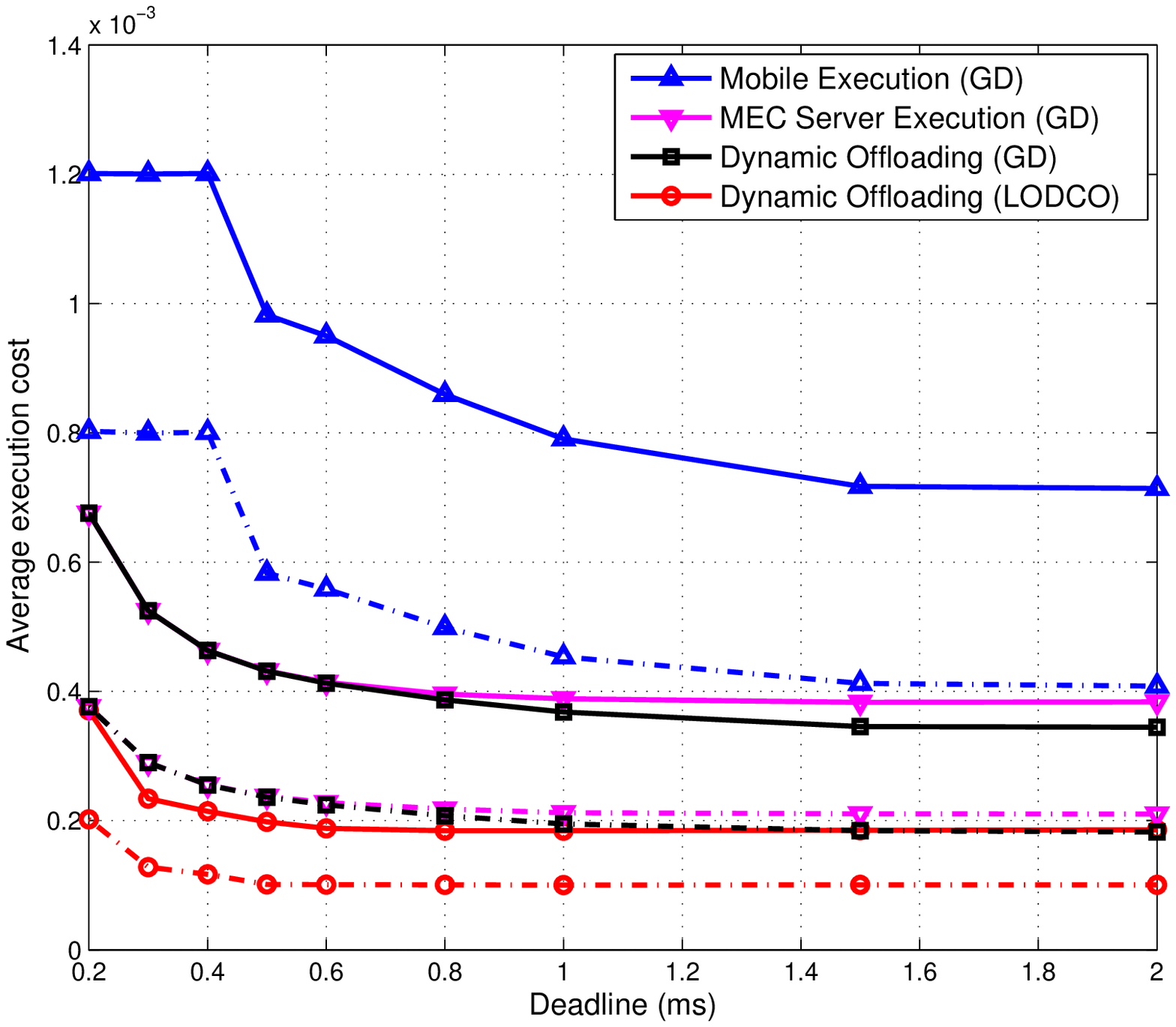}} }  %Case1.eps
  \hspace{-25pt}
  \subfigure[Average completion time/task drop ratio vs. $\tau_{d}$]{
    \label{CTDRDealine} %% label for second subfigure
    \raisebox{-1cm}{\includegraphics[width=0.5\textwidth]{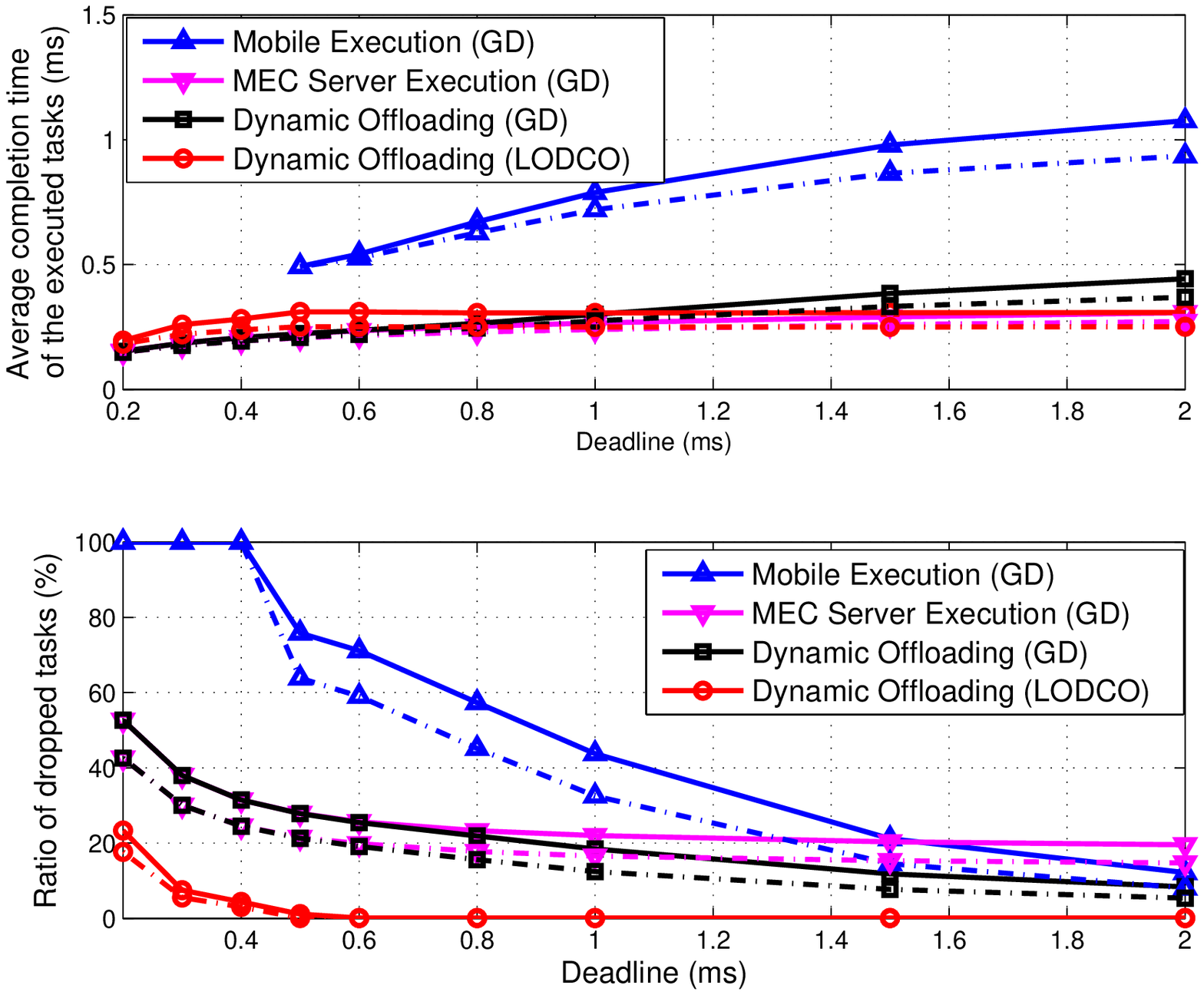}}}     %Case2.eps
  \vspace{-5pt}
  \caption{System performance vs. execution deadline, the solid lines corresponds to $\rho=0.6$ and the dash-solid lines corresponds to $\rho=0.4$.}
  \label{ImpactDeadline} %% label for entire figure
\end{figure}

In Fig. \ref{ImpactDeadline}, we reveal the relationship of between the execution deadline $\tau_{d}$ and the system performance. As $\tau_{d}$ decreases, i.e., the computation requirement becomes more stringent, the execution cost, average completion time and task drop ratio achieved by all four policies increase. It can be seen that when $\tau_{d}\leq 0.4$ ms, the execution cost achieved by the Mobile Execution (GD) policy becomes a constant $\rho \phi$, and the task drop ratio is 100\%. Meanwhile, the MEC Server Execution (GD) and the Dynamic Offloading (GD) policies converge. In these scenarios, the mobile device is not able to conduct any computation because of hardware limitation, i.e., $f^{t}\leq f_{\rm{CPU}}^{\max}=1.5$ GHz, and all the computation tasks have to be offloaded to the MEC server for MEC. The results in Fig. \ref{CTDRDealine} confirms the benefits of MEC as around 50\% tasks are successfully executed for $\tau_{d}=0.2$ ms even under the greedy offloading policy. Note that for a small value of $\tau_{d}$, e.g., $\tau_{d}\leq 0.8$ ms, the average completion time achieved by the LODCO algorithm is slightly longer than those of the other two policies with computation offloading, but the task drop ratio is reduced noticeably by more than 20\%. This phenomenon is similar to what was observed in Fig. \ref{CTDRRho}, where the LODCO algorithm tends to avoid dropping tasks by prolonging the average completion time in order to achieve a minimum execution cost.
\vspace{-10pt}
\begin{figure}[!htbp]
  \centering
  \subfigure[Execution cost vs. $d$]{
    \label{ECDistance} %% label for first subfigure
    \raisebox{-1cm}{\includegraphics[width=0.5\textwidth]{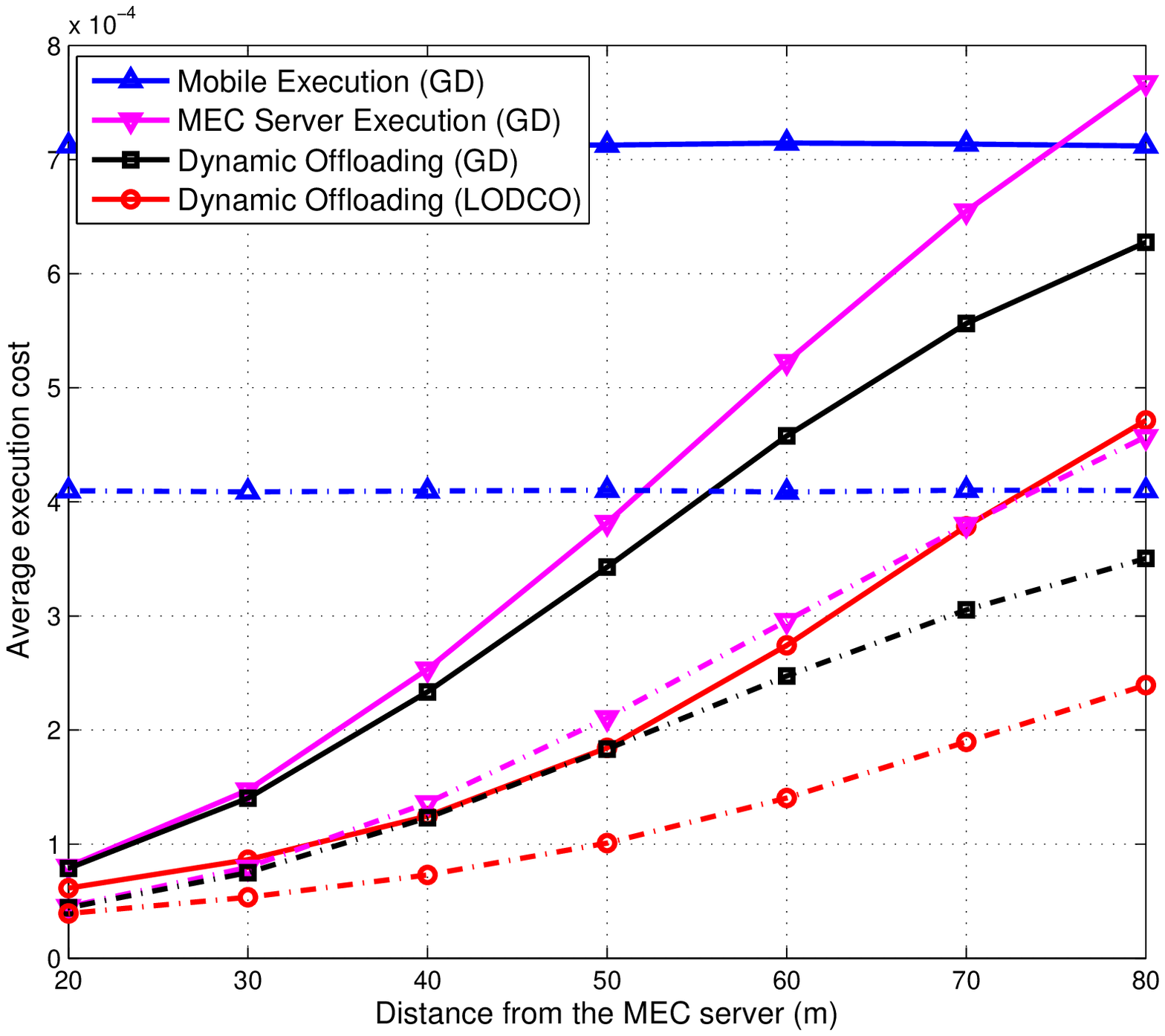}} }  %Case1.eps
  \hspace{-25pt}
  \subfigure[Average completion time/task drop ratio vs. $d$]{
    \label{CTDRDistance} %% label for second subfigure
    \raisebox{-1cm}{\includegraphics[width=0.5\textwidth]{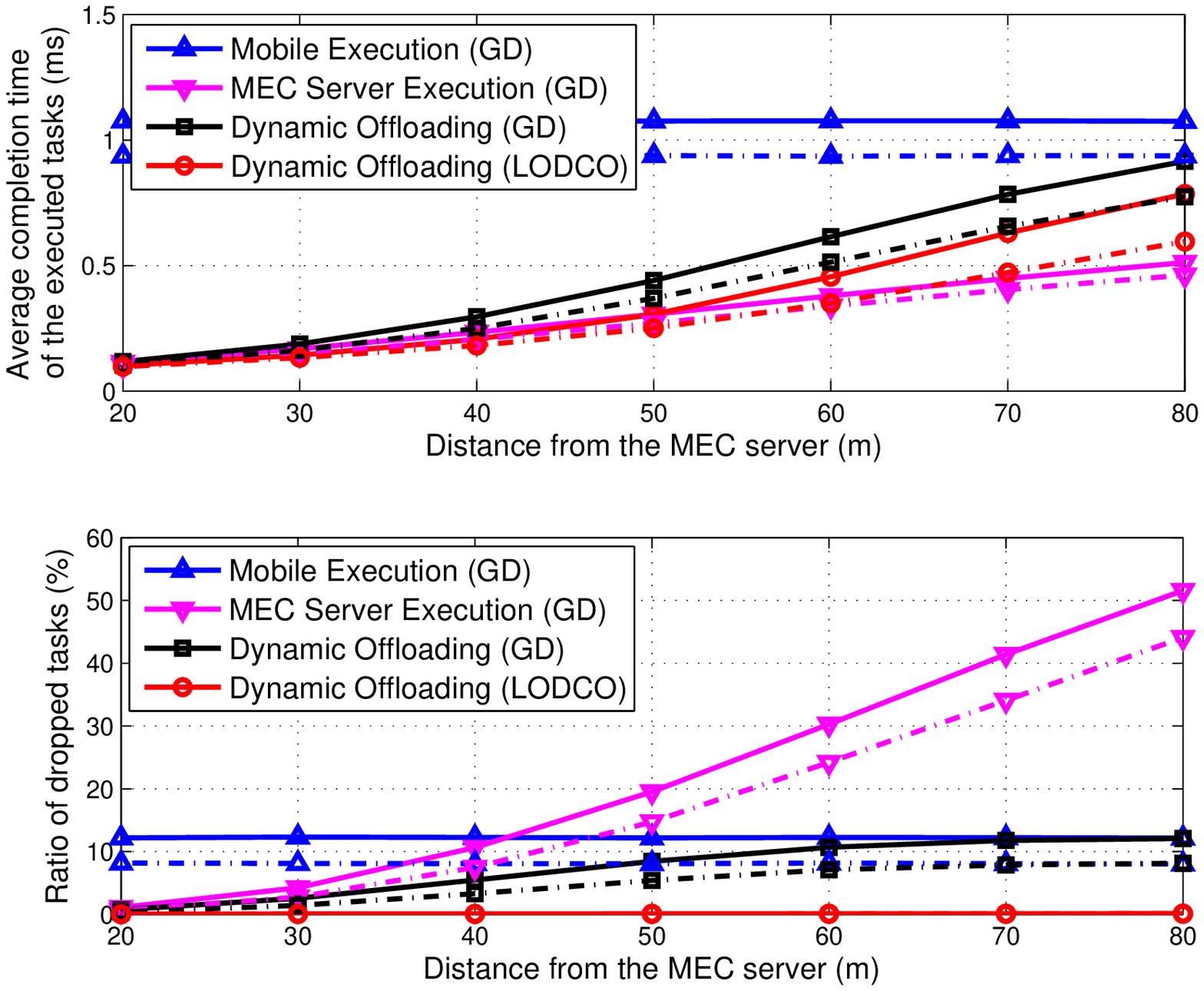}}}     %Case2.eps
  \vspace{-5pt}
  \caption{System performance vs. distance, the solid lines corresponds to $\rho=0.6$ and the dash-solid lines corresponds to $\rho=0.4$.}
  \label{ImpactDistance} %% label for entire figure
\end{figure}

Finally, we show the relationship between the system performance and $d$, i.e., the distance from the mobile device to the MEC server, in Fig. \ref{ImpactDistance}. The performance of the computation offloading policies, including the MEC Server Execution (GD), Dynamic offloading (GD) and the LODCO algorithms, deteriorate as $d$ becomes large. As can be seen from Fig. \ref{ECDistance}, when the mobile device is close to the MEC server, the three computation offloading policies converge and greatly outperform the Mobile Execution (GD) policy. In such scenarios, the mobile device is able to offload the computation
tasks to the MEC server with a small amount of harvested energy due to small path loss. With a large value of $d$, e.g., $d =80$ m, offloading the tasks greedily cannot bring any execution cost reduction compared the Mobile Execution (GD) policy, while the LODCO algorithm offers more than 40\% performance gain. From Fig. \ref{CTDRDistance}, we see that although the MEC Server Execution (GD) policy incurs the least completion time for the executed tasks, its task failure performance sharply degrades. In contrast, the proposed LODCO algorithm achieves a near-zero task drop ratio with an improved completion time performance compared to the Mobile Execution (GD) and Dynamic Offloading (GD) policies.

\section{Conclusions}
In this paper, we investigated mobile-edge computing (MEC) systems with EH mobile devices. The execution cost, which addresses the execution delay and task failure, was adopted as the performance metric. A dynamic computation offloading policy, namely, the Lyapunov optimization-based dynamic computation offloading (LODCO) algorithm, was then developed. It is a low-complexity online algorithm and requires little prior knowledge. We found the monotonic properties of the CPU-cycle frequencies (transmit power) for mobile execution (computation offloading) with respect to the battery energy level, which uncovers the impact of EH to the system operations. Performance analysis was conducted which revealed the asymptotic optimality of the proposed algorithm. Simulation results showed that the proposed LODCO algorithm not only significantly outperforms the benchmark greedy policies in terms of execution cost, but also reduces computation failures noticeably at an expense of minor execution delay performance degradation. Our study provides a viable approach to design future MEC systems with renewable energy-powered devices. {It would be interesting to extend the proposed algorithm to more general MEC systems with multiple mobile devices, as well as consider resource-limited MEC servers. Another extension is to combine the concepts of wireless energy transfer and energy harvesting by deploying a power beacon co-located with the MEC server so that the energy deficit incurred by the renewable energy sources can be compensated by the controllable radio frequency energy.}

%%%%%%%%%%%%%%%%%%%%%%%%%%%%%%%%%%%%%%%%%%%%%%%%%%%%%%%
%%%%%%%%%%%%%%%%%%%%%%%%%%%%%%%%%%%%%%%%%%%%%%%%%%%%%%%

\begin{appendix}
\subsection{Proof for Proposition \ref{tightperf}}
Since $\mathcal{P}_{2}$ is a tightened version of $\mathcal{P}_{1}$, we have $\rm{EC}_{\mathcal{P}_{1}}^{*}\leq \rm{EC}_{\mathcal{P}_{2}}^{*}$. The other side of the inequality can be obtained by constructing a feasible solution for $\mathcal{P}_{2}$ (denoted as $\langle e^{t}_{\mathcal{P}_{2}}, \bm{I}_{\mathcal{P}_{2}}^{t},f^{t}_{\mathcal{P}_{2}},p^{t}_{\mathcal{P}_{2}}\rangle$) based on the optimal solution for $\mathcal{P}_{1}$ (denoted as $\langle e^{t}_{\mathcal{P}_{1}}, \bm{I}_{\mathcal{P}_{1}}^{t},f^{t}_{\mathcal{P}_{1}},p^{t}_{\mathcal{P}_{1}}\rangle$\footnote{For simplicity, we assume the optimal solution for $\mathcal{P}_{1}$ satisfies the property of the optimal CPU-cycle frequencies in Lemma \ref{equalfreq}.}): \textbf{i)} If $\mathcal{E}\left(\bm{I}^{t}_{\mathcal{P}_{1}},f_{\mathcal{P}_{1}}^{t},p_{\mathcal{P}_{1}}^{t}\right)\in\left(0,E_{\min}\right)$, then the computation task will be dropped in the constructed solution and no harvested energy will be consumed, i.e., ${\rm{cost}}^{t}_{\mathcal{P}_{2}}=\phi$; \textbf{ii)} If $\mathcal{E}\left(\bm{I}^{t}_{\mathcal{P}_{1}},f_{\mathcal{P}_{1}}^{t},p_{\mathcal{P}_{1}}^{t}\right)\in\left[E_{\min},E_{\max}\right]$, the constructed solution for the $t$th time slot will be the same as the optimal solution for $\mathcal{P}_{1}$; \textbf{iii)} The EH decision $e_{\mathcal{P}_{2}}^{t}$ is determined by
%\begin{equation}
$e_{\mathcal{P}_{2}}^{t}=\max\{B^{t}_{\mathcal{P}_{1}}-\mathcal{E}\left(\bm{I}^{t}_{\mathcal{P}_{1}},f_{\mathcal{P}_{1}}^{t},p_{\mathcal{P}_{1}}^{t}\right)
+e_{\mathcal{P}_{1}}^{t}-B_{\mathcal{P}_{2}}^{t}+\mathcal{E}\left(\bm{I}^{t}_{\mathcal{P}_{2}},f_{\mathcal{P}_{2}}^{t},p_{\mathcal{P}_{2}}^{t}\right),0\}$.
%\begin{itemize}
%\item If $\mathcal{E}\left(\bm{I}^{t}_{\mathcal{P}_{1}},f_{\mathcal{P}_{1}}^{t},p_{\mathcal{P}_{1}}^{t}\right)\in\left(0,E_{\min}\right)$, then the %computing task will be dropped in the constructed solution, i.e., $I_{\rm{d},\mathcal{P}_{2}}^{t}=1$, %$I_{\rm{m},\mathcal{P}_{2}}^{t}=I_{\rm{s},\mathcal{P}_{2}}^{t}=0$, and no harvested energy will be consumed and %${\rm{cost}}^{t}_{\mathcal{P}_{2}}=\phi$.
%\item If $\mathcal{E}\left(\bm{I}^{t}_{\mathcal{P}_{1}},f_{\mathcal{P}_{1}}^{t},p_{\mathcal{P}_{1}}^{t}\right)\in\left[E_{\min},E_{\max}\right]$, the %constructed solution for the $t$th time slot will be the same as the optimal solution for $\mathcal{P}_{1}$, i.e., %$\bm{I}^{t}_{\mathcal{P}_{2}}=\bm{I}^{t}_{\mathcal{P}_{1}}$, $f_{\mathcal{P}_{1}}^{t}=f_{\mathcal{P}_{2}}^{t}$, and %$p_{\mathcal{P}_{1}}^{t}=p_{\mathcal{P}_{2}}^{t}$.
%\item The energy harvesting decision $e_{\mathcal{P}_{2}}^{t}$ is determined by
%%\begin{equation}
%$e_{\mathcal{P}_{2}}^{t}=\max\{B^{t}_{\mathcal{P}_{1}}-\mathcal{E}\left(\bm{I}^{t}_{\mathcal{P}_{1}},f_{\mathcal{P}_{1}}^{t},p_{\mathcal{P}_{1}}^{t}\right)
%+e_{\mathcal{P}_{1}}^{t}-B_{\mathcal{P}_{2}}^{t}+\mathcal{E}\left(\bm{I}^{t}_{\mathcal{P}_{2}},f_{\mathcal{P}_{2}}^{t},p_{\mathcal{P}_{2}}^{t}\right),0\}.$
%%\end{equation}
%\end{itemize}

It is not difficult to show $B_{\mathcal{P}_{1}}^{t}\leq B_{\mathcal{P}_{2}}^{t}<+\infty$, and thus the constructed solution is feasible to $\mathcal{P}_{2}$. If $E_{\min}\geq E_{\min}^{\tau_{d}}$, where $E_{\min}^{\tau}=\kappa W^{3}\tau_{d}^{-2}$ is the minimum amount of energy required to meet the deadline constraint for mobile execution, for a time slot with $I_{\rm{m},\mathcal{P}_{1}}^{t}=1$ and
$\mathcal{E}\left(\bm{I}^{t}_{\mathcal{P}_{1}},f_{\mathcal{P}_{1}}^{t},p_{\mathcal{P}_{1}}^{t}\right)\in\left(0,E_{\min}\right)$, the constructed solution incurs $\left(\phi-\tau_{E_{\min}}\right)$ units of extra execution cost in the worst case. Here, $\tau_{E_{\min}}=\kappa^{\frac{1}{2}}W^{\frac{3}{2}}E_{\min}^{-\frac{1}{2}}$ is the execution delay corresponds to $E_{\min}$ amount of energy consumption for mobile execution; otherwise, if $E_{\min}< E_{\min}^{\tau_{d}}$, $I_{\rm{m},\mathcal{P}_{1}}^{t}=1$ and
$\mathcal{E}\left(\bm{I}^{t}_{\mathcal{P}_{1}},f_{\mathcal{P}_{1}}^{t},p_{\mathcal{P}_{1}}^{t}\right)\in\left(0,E_{\min}\right)$ is infeasible as the deadline constraint cannot be met. Besides, the probability of offloading a task to the MEC server successfully with energy consumption less than $E_{\min}$ is no greater than $\mathbb{P}\{\omega \tau_{d} \log_{2}\left(1+\frac{h^{t}p^{t}}{\sigma}\right)\geq L\}=1-F_{H}\left(\eta\right)$, where $\eta \triangleq \left(2^{\frac{L}{\omega\tau_{d}}}-1\right)\tau_{d}\sigma E_{\min}^{-1}$, and the constructed solution will incur at most $\phi$ units of extra execution cost as ${\rm{cost}}_{\mathcal{P}_{1}}^{t}>0$. By further incorporating the task request probability $\rho$, we can obtain the desired result.

\subsection{Proof for Corollary \ref{property2}}
For $\tilde{B}^{t}<0$, since $\Xi\left(h^{t},p_{0}^{t},\tilde{B}^{t}\right)=0$, with some manipulations,
%\begin{equation}
%\tilde{B}^{t}\left[\frac{h^{t}p_{0}^{t}}{\left(h^{t}p_{0}^{t}+\sigma\right)\ln %2}-\log_{2}\left(1+\frac{h^{t}p_{0}^{t}}{\sigma}\right)\right]=\frac{Vh^{t}}{\left(\sigma + h^{t}p_{0}^{t}\right)\ln 2}
%\label{proofproperty2eq1}
%\end{equation}
we have $\tilde{B}^{t} \cdot k\left(h^{t},p_{0}^{t}\right)=\frac{h^{t}V}{\ln 2}$, where
%\begin{equation}
$k\left(h,p\right)=\frac{hp}{\ln 2}-\left(hp+\sigma\right)\log_{2}\left(1+\frac{hp}{\sigma}\right)$,
%\end{equation}
and
%\begin{equation}
$\frac{\partial k\left(h,p\right)}{\partial p}=-h\log_{2}\left(1+\frac{hp}{\sigma}\right)<0$,
%\end{equation}
i.e., $k\left(h,p\right)$ decreases with $p$ for $p > 0$. Denote $\tilde{B}^{t}_{-}<\tilde{B}^{t}_{+}<0$ and the corresponding solutions
for $\Xi\left(h^{t},p,\tilde{B}^{t}\right)=0$ as $p_{0,-}^{t}$ and $p_{0,+}^{t}$, respectively. Since $\tilde{B}^{t}_{+}k\left(h^{t},p^{t}_{0,+}\right)=\tilde{B}^{t}_{-}k\left(h^{t},p^{t}_{0,-}\right)>0$, we have $k\left(h^{t},p^{t}_{0,+}\right)<k\left(h^{t},p^{t}_{0,-}\right)<0$, i.e., $p_{0,+}^{t}>p_{0,-}^{t}$. Since $p_{L}^{t}$ and $p_{U}^{t}$ are invariant with $\tilde{B}^{t}$, according to (\ref{optOLpwr}), $p^{t*}$ is non-decreasing with $\tilde{B}^{t}$ for $\tilde{B}^{t}<0$. Besides, as $p^{t*}=p^{t}_{U}$ when $\tilde{B}^{t}>0$, we can conclude that $p^{t*}$ is non-decreasing with $\tilde{B}^{t}$.

\subsection{Proof for Lemma \ref{lmaupperboundlyaV}}
By subtracting $\theta$ at both sides of (\ref{batterydynamics}), we have
%\begin{equation}
$\tilde{B}^{t+1}=\tilde{B}^{t}+e^{t}-\mathcal{E}\left(\bm{I}^{t},f^{t},p^{t}\right)$.
%\label{pertbatterydynamics}
%\end{equation}
Squaring both sides of this equality, we have
\begin{equation}
\begin{split}
\left(\tilde{B}^{t+1}\right)^{2}&=\left(\tilde{B}^{t}+e^{t}-\mathcal{E}\left(\bm{I}^{t},f^{t},p^{t}\right)\right)^{2}\\
%&=\left(\tilde{B}^{t}\right)^{2}+2\tilde{B}^{t}\left(e^{t}-\mathcal{E}\left(\bm{I}^{t},f^{t},p^{t}\right)\right)+\left(e^{t}-\mathcal{E}\left(\bm{I}^{t},f^{t},p^{t}\right)\right)^{2}\\
& \leq \left(\tilde{B}^{t}\right)^{2} + 2\tilde{B}^{t}\left(e^{t}-\mathcal{E}\left(\bm{I}^{t},f^{t},p^{t}\right)\right) + \left(e^{t}\right)^{2} +\mathcal{E}^{2}\left(\bm{I}^{t},f^{t},p^{t}\right)\\
& \leq \left(\tilde{B}^{t}\right)^{2} + 2\tilde{B}^{t}\left(e^{t}-\mathcal{E}\left(\bm{I}^{t},f^{t},p^{t}\right)\right) + \left(E_{H}^{\max}\right)^{2} + \tilde{E}_{\max}^{2}.
\end{split}
\label{lyabound1}
\end{equation}
Dividing both sides of (\ref{lyabound1}) by $2$, adding $V\left[\mathcal{D}\left(\bm{I}^{t},f^{t},p^{t}\right)+\phi\cdot \bm{1}\left(\zeta^{t}=1,I_{\rm{d}}^{t}=1\right)\right]$, as well as taking the expectation conditioned on $\tilde{B}^{t}$, we can obtain the desired result.
\end{appendix}

\subsection{Proof for Theorem \ref{thmasymopt}}
Since the LODCO algorithm obtains the optimal solution of the per-time slot problem, (\ref{ProofThm1bound}) holds, where ${\rm{cost}}^{t*}$ and ${\rm{cost}}^{t\Pi}$ are the execution cost at the $t$th time slot under $\langle \bm{I}^{t*},f^{t*},p^{t*}\rangle$ and $\langle \bm{I}^{t\Pi},f^{t\Pi},p^{t\Pi}\rangle$, respectively. ($\dagger$) is because that policy $\Pi$ is independent of the battery energy level $B^{t}$, and ($\ddagger$) is due to Corollary \ref{batterybound} and Lemma \ref{lmaarbclose}.
\begin{equation}
\begin{split}
\Delta_{V}\left(\tilde{B}^{t}\right)
&\leq \mathbb{E}\bigg[\tilde{B}^{t}\left[e^{t*}-\mathcal{E}\left(\bm{I}^{t*},f^{t*},p^{t*}\right)\right]+V\cdot {\rm{cost}}^{t*}|\tilde{B}^{t}\bigg]+C\\
&\leq \mathbb{E}\bigg[\tilde{B}^{t}\left[e^{t\Pi}-\mathcal{E}\left(\bm{I}^{t\Pi},f^{t\Pi},p^{t\Pi}\right)\right]+V\cdot {\rm{cost}}^{t\Pi}|\tilde{B}^{t}\bigg]+C\\
&\mathop{=}^{(\dagger)} \tilde{B}^{t}\mathbb{E}\left[e^{t\Pi}-\mathcal{E}\left(\bm{I}^{t\Pi},f^{t\Pi},p^{t\Pi}\right)\right]+V\cdot \mathbb{E}\left[{\rm{cost}}^{t\Pi}\right]+C\\
&\mathop{\leq}^{(\ddagger)}\max\{\theta,E_{H}^{\max}\}\cdot \varrho \delta + V \left({\rm{EC}_{\mathcal{P}_{3}}}+\delta\right)+C.
\end{split}
\label{ProofThm1bound}
\end{equation}
%\begin{equation}
%\begin{split}
%&\Delta_{V}\left(\tilde{B}^{t}\right)\\
%&\leq %\mathbb{E}\bigg[\tilde{B}^{t}\left[e^{t*}-\mathcal{E}\left(\bm{I}^{t*},f^{t*},p^{t*}\right)\right]+V\left[\mathcal{D}\left(\bm{I}^{t*},f^{t*},p^{t*}\right)+
%\phi\cdot\bm{1}\left(\zeta^{t}=1,I_{\rm{d}}^{t*}=1\right)\right]|\tilde{B}^{t}\bigg]+C\\
%&\leq %\mathbb{E}\bigg[\tilde{B}^{t}\left[e^{t\Pi}-\mathcal{E}\left(\bm{I}^{t\Pi},f^{t\Pi},p^{t\Pi}\right)\right]+V\left[\mathcal{D}\left(\bm{I}^{t\Pi},f^{t\Pi},p^{t\Pi}\right)+
%\phi\cdot\bm{1}\left(\zeta^{t}=1,I_{\rm{d}}^{t\Pi}=1\right)\right]|\tilde{B}^{t}\bigg]+C\\
%&\mathop{=}^{(\dagger)} %\tilde{B}^{t}\mathbb{E}\left[e^{t\Pi}-\mathcal{E}\left(\bm{I}^{t\Pi},f^{t\Pi},p^{t\Pi}\right)\right]+V\mathbb{E}\left[\mathcal{D}\left(\bm{I}^{t\Pi},f^{t\Pi},p^{t\Pi}\right)+\phi\cdot %\bm{1}\left(\zeta^{t}=1,I_{\rm{d}}^{t\Pi}=1\right)\right]+C\\
%&\mathop{\leq}^{(\ddagger)}\max\{\theta,E_{H}^{\max}\}\cdot \varrho \delta + V \left({\rm{EC}_{\mathcal{P}_{3}}}+\delta\right)+C,
%\end{split}
%\end{equation}
By letting $\delta$ go to zero, we obtain
\begin{equation}
\Delta_{V}\left(\tilde{B}^{t}\right)\leq V {\rm{EC}_{\mathcal{P}_{3}}^{*}} + C.
\label{asym2}
\end{equation}
Taking the expectation on both sides of (\ref{asym2}), summing up the inequalities for $t=0,\cdots T-1$, dividing by $T$ and letting $T$ go to infinity, we have ${\rm{EC}}_{\rm{LODCO}}\leq {\rm{EC}}_{\mathcal{P}_{3}}^{*}+\frac{C}{V}$. By further utilizing Proposition \ref{tightperf} and Lemma \ref{relaxProb}, the theorem is proved.

\end{document}